\documentclass[9pt,shortpaper,twoside,web]{ieeecolor}
\usepackage{generic}

\let\labelindent\relax
\usepackage{enumitem}

\usepackage{amsthm}

\newtheorem{problem}{Problem}
\newtheorem{assumption}{Assumption}
\newtheorem{remark}{Remark}

\newtheorem{lemma}{Lemma}
\newtheorem{theorem}{Theorem}

\usepackage{cite}
\usepackage{amsmath,amssymb,amsfonts}
\usepackage{mathtools}

\usepackage{algorithmic}
\usepackage{algorithm}

\usepackage{graphicx}
\usepackage{textcomp}
\usepackage{xcolor}
\usepackage{comment}

\usepackage{tikz}
\usepackage[american]{circuitikz}
\usetikzlibrary{positioning, arrows.meta, calc}

\def\BibTeX{{\rm B\kern-.05em{\sc i\kern-.025em b}\kern-.08em
    T\kern-.1667em\lower.7ex\hbox{E}\kern-.125emX}}
\begin{document}
\title{A Gray-Box Approach for \\ Decentralized Grid-Equivalent Model Identification}
\author{Sanjay Chandrasekaran$^{1,2}$,  Florian Dörfler$^{2}$ and  Silvia Mastellone$^{1}$
\thanks{This work was supported by the Swiss National Science Foundation under NCCR Automation, grant agreement 51NF40\_180545}
\thanks{$^{1}$Sanjay Chandrasekaran and Silvia Mastellone are with the Institute of Electric Power Systems, University of Applied Sciences Northwest Switzerland, Windisch, Switzerland. \texttt{ sanjay.chandrasekaran@fhnw.ch, silvia.mastellone@fhnw.ch}}
\thanks{$^{2}$Sanjay Chandrasekaran and Florian Dorfler are with the Automatic Control Lab (IfA) at ETH Zurich, Switzerland. 
\texttt{schandraseka@ethz.ch, dorfler@ethz.ch}}}

\maketitle

\begin{abstract}

We propose a decentralized, frequency-domain identification algorithm that estimates the grid-equivalent model using local measurements from the perspective of each converter. Since local electric signals in a multi-converter setup are affected by voltage inputs from the grid, estimating a direct equivalent impedance yields biased and inaccurate results. To overcome this, we design a framework that decouples the effect of the equivalent impedance (passive) from that of the equivalent voltage (active). The parameters and equivalent grid voltages are then estimated using a least-squares algorithm and a Kalman filter, respectively, applied across frequency samples, with additional pre-processing techniques to remove the influence of the grid on the local estimated models. We then demonstrate the accuracy and performance of our algorithm on an interconnected $5-$converter system in grid-forming mode.
\end{abstract}


\section{Introduction}
\label{sec:introduction}
Over the past few decades, climate change and sustainability concerns have motivated a remarkable increase in the use of renewable energy sources and associated power electronic technologies within the electric grid \cite{milano}. Some persistent challenges with the integration of power electronics into the grid have been data-driven modeling, control, and stability analysis of the grid \cite{milano,gross,uros}. The estimation and/or knowledge of grid impedance, for example, has aided in ascertaining small-signal stability \cite{sun,impedancebasedstability}, enabling adaptive control to ensure stability \cite{adaptivecontrol,adaptivecontrol1}, and ensuring the safety of distributed generation systems by identifying islanding conditions  \cite{islanding}.

To identify the grid-impedance model, the grid (ignoring individual device dynamics) is represented as a "small-signal model" \cite{uros,decentralizedstability,mimoidentificationverena}, a linearized transfer function model that represents the equivalent impedance relating the terminal voltage as inputs and current injected as outputs at points of common couplings (PCCs). However, identifying the grid impedance model for large-scale power grids in a centralized manner presents several challenges; the primary one being the difference in time-scales between low-inertia renewables and high-inertia synchronous generators, thus making it difficult to synchronize measurement time-stamps \cite{grideqest}, along with scalability issues. Hence, we are interested in estimating the small-signal equivalent impedance/ admittance model from the perspective of local converters at each PCC, along with the influence of the rest of the grid on the local converters, in a decentralized and parallel manner.
\begin{figure}
    \centering
    \scalebox{0.55}{\begin{tikzpicture}[
        >=stealth,
        networkNode/.style={circle, draw=gray, thick, fill=gray!10, inner sep=1.5pt},
        myblue/.style={color=cyan!60!blue}
    ]
    
    \coordinate (posN1)  at (0, 2.5);
    \coordinate (posN2)  at (-1, 0.5);
    \coordinate (posNi)  at (3.5, -1.8);
    \coordinate (posNn)  at (6, 2);
    \coordinate (posNn1) at (6.2, 0);
    
    \draw[myblue, dashed, semithick] plot [smooth cycle, tension=0.8] coordinates {
        (posN1) (posNn) (posNn1) (posNi) (posN2)
    };
    
    \node[myblue, font=\small] at (3, 3) {Interconnected grid};
    
    \draw[gray, semithick] (posN1) -- (posN2);
    \draw[gray, semithick] (posN1) -- (posNi);
    \draw[gray, semithick] (posN1) -- (posNn);
    \draw[gray, semithick] (posN2) -- (posNi);
    \draw[gray, semithick] (posNn) -- (posNi);
    \draw[gray, semithick] (posNn) -- (posNn1);
    \draw[gray, semithick] (posNn1) -- (posNi);
    
    \node[networkNode, label={above:1}] (N1) at (posN1) {};
    \node[networkNode, label={above left:2}] (N2) at (posN2) {};
    \node[networkNode, label={below:$i$}] (Ni) at (posNi) {};
    \node[networkNode, label={above right:$N$}] (Nn) at (posNn) {};
    \node[networkNode, label={below right:$N-1$}] (Nn1) at (posNn1) {};
    
    \node at (2.8, 0.7) {$\cdot$};
    \node at (3.1, 0.4) {$\cdot$};
    \node at (3.4, 0.1) {$\cdot$};
    
    \newcommand{\drawconverter}[3]{
        \node[draw, rounded corners=4pt, minimum width=1.4cm, minimum height=1.1cm, semithick] (#3) at (#1,#2) {};
        \begin{scope}[shift={(#3.center)}, scale=0.7]
            \draw[semithick] (-0.3, 0.3) -- (-0.3, -0.3);
            \draw[semithick] (-0.15, 0.3) -- (-0.15, -0.3);
            \draw[semithick] (-0.15, 0.2) -- (0.1, 0.4) -- (0.1, 0.5);
            \draw[semithick] (-0.15, -0.2) -- (0.1, -0.4) -- (0.1, -0.5);
            \draw[->, >=stealth, semithick] (-0.15, -0.2) -- (-0.025, -0.3);
            \draw[semithick] (0.1, 0.4) -- (0.4, 0.4) -- (0.4, -0.4) -- (0.1, -0.4);
            \draw[semithick] (0.25, 0.2) -- (0.55, 0.2);
            \draw[semithick, fill=white] (0.25, -0.15) -- (0.55, -0.15) -- (0.4, 0.2) -- cycle;
        \end{scope}
    }
    
    \drawconverter{-2.5}{2.5}{C1}
    \drawconverter{-3}{0.5}{C2}
    \drawconverter{-0.5}{-1.8}{Ci}
    \drawconverter{8}{2}{Cn}
    \drawconverter{8.5}{0}{Cn1}
    
    \node[below=2pt of C1] {VSC $1$};
    \node[below=2pt of C2] {VSC $2$};
    \node[below=2pt of Ci] {VSC $i$};
    \node[below=2pt of Cn1] {VSC $N-1$};
    \node[below=2pt of Cn] {VSC $N$};
    
    \draw[semithick] (C1.east) -- (N1.west);
    \draw[semithick] (C2.east) -- (N2.west);
    \draw[semithick] (Ci.east) -- (Ni.west);
    \draw[semithick] (Cn.west) -- (Nn.east);
    \draw[semithick] (Cn1.west) -- (Nn1.east);
    
    \node at (-1.5, -0.4) {$\cdot$};
    \node at (-1.2, -0.7) {$\cdot$};
    \node at (-0.9, -1.0) {$\cdot$};
    
    \draw[-{Latex[length=2mm, width=1.5mm]}, myblue, thick] 
        (-1.5, 2.7) -- (-0.2, 2.7) 
        node[midway, above, text=black, inner sep=2pt] {$i_{dq,1}$};
        
    \draw[-{Latex[length=2mm, width=1.5mm]}, myblue, thick] 
        (-0.2, 2.7) -- (-0.2, 3.5) 
        node[above, text=black, inner sep=2pt] {$v_{dq,1}$};
        
    \end{tikzpicture}}
    
    \caption{Illustration of the interconnection between multiple VSCs \cite{decentralizedstability}.}
    \label{fig:nconvertersystem}
\end{figure}
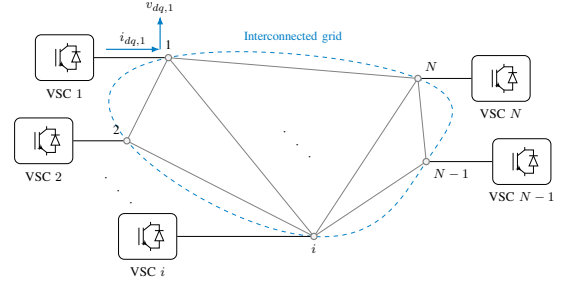

\subsection*{Related Work}

A comprehensive overview of system identification (SysID) algorithms for determining the equivalent grid impedance transfer function is provided in \cite{mimoidentificationverena}. In the class of non-parametric methods, "frequency sweep" methods are used \cite{sweep1,sweep2}, wherein multiple sinusoidal perturbation signals are introduced, after which, the impedance is reconstructed using the recorded currents and voltages at different frequencies. Unfortunately, this approach would not only disrupt normal  grid operation but also be time-consuming. In \cite{ultrafast}, a fast non-parametric algorithm was developed that identifies the grid impedance in $dq$ frame by converting it from a multi-input multi-output (MIMO) transfer function to two single-input single-output (SISO) transfer functions. In the domain of parametric methods, conventional prediction-error methods like auto-regressive methods with exogenous inputs (ARMAX), and sub-space algorithms that estimate the state-space equivalent, are used \cite{Ljung1998}. In \cite{amin}, a recursive least-squares algorithm in the frequency domain is used to determine the system parameters. 

As far as the multi-agent setup is concerned, \cite{communicationfree} discusses the influence of the grid on local converters, particularly when voltage source converters (VSCs) carry out simultaneous excitation. To counter this, they develop a "communication-free" algorithm by establishing a trigger mechanism embedded into each VSC that enforces excitations only during particular time intervals. In \cite{grideqest}, the equivalent grid impedance and Thevenin voltage are estimated in a decentralized manner. However, the estimation involves collecting data at multiple active and reactive power operating points, thus disrupting grid operation. In \cite{mininvasiveEKF}, an extended Kalman filter was designed to estimate the equivalent grid impedance and real-time Thevenin voltage. However, it does
not discuss the scenario of simultaneous excitations. In \cite{PandO}, a technique "\emph{perceive and optimize}" (P\&O) was developed. In this method, the first step is to "\emph{perceive}": estimate local grid dynamics in the form a parameterized MIMO transfer function, which is carried out using the algorithms described in \cite{mimoidentificationverena}. The second step, i.e., "\emph{optimize}", is to define an optimal control problem, optimizing over a set of parameters describing the ideal closed-loop MIMO transfer function adhering to industry-standard grid codes. To account for the variability of grid operating conditions, multiple cycles of P\&O are carried out in equal intervals. Although the P\&O algorithm works efficiently in the single-converter case, it requires coordination between each converter in the multi-agent setup.

\subsection*{Contributions}
While existing decentralized estimation algorithms can theoretically operate in parallel at each point of common coupling (PCC), they rely on restrictive operational assumptions. Specifically, they require coordinated system identification (SysID) \cite{communicationfree,PandO,onlinegridimpedanceestimation}, multiple operating points \cite{grideqest,sweep1}, not considering closed-loop control interactions \cite{amin}, or the unrealistic assumption of a perfectly stiff equivalent voltage (i.e., an infinite bus) \cite{mininvasiveEKF,mimoidentificationverena,ultrafast,amin}. These assumptions frequently break down in dynamic, multi-agent networks.

To overcome these limitations and enable robust identification under active grid conditions, we propose a gray-box identification algorithm with the following primary contributions:
\begin{enumerate}
    \item \textbf{Fully decentralized execution}: Our algorithm operates independently from each VSC using strictly local PCC measurements, requiring zero communication or coordinated timing with the rest of the grid. This inherent decoupling makes the approach highly scalable to grids of arbitrary sources and/ or loads. 
    \item \textbf{Separating active \& passive entities}: From the perspective of the local VSCs, the algorithm extracts the parametric equivalent grid admittance and the non-parametric equivalent grid voltage. We separate the active and passive components so that certain crucial physical information (for example, the equivalent inertia and damping of the grid) can be extracted. By executing this estimation in the frequency domain and using instrumental variables (IV), it is possible to decouple the physical transmission line dynamics from the active control reactions of other VSCs and loads, thus resolving the severe cross-coupling errors that corrupt time-domain identification.
    
    \item \textbf{Guarantees \& Robustness}: We are able to provide theoretical guarantees on the equivalent admittance and voltage estimation errors and offer inferences that could be useful in further reducing the errors. Simulation results validate the algorithm's accuracy using continuous, minimally invasive, wide-band excitations. Our method accurately identifies grid parameters despite non-ideal conditions, i.e., a low-inertia grid (not an infinite bus), heterogeneous lines, multi-agent simultaneous excitations, and withstanding deviations in nominal amplitude and frequency.
\end{enumerate}

\subsection*{Outline}
The rest of this paper is structured as follows: Section \ref{problemsetupsection} describes the modeling and the formal problem statement. Section \ref{algorithmsection} describes the measurements recorded, grid-equivalent map, frequency discrimination (FD), and the algorithm to identify the grid-equivalent model. We show the simulation results on an interconnected $5-$converter network in Section \ref{resultssection}. Section \ref{conclusionsection} concludes our work and presents potential future directions.
\subsection*{Notations}
We denote $s=\mathbf{j}\omega$ as the transfer function variable. We denote $\mathbb{C}^{n\times n}$ to be an $n\times n$ complex matrix, and $\mathbb{R}_{\geq 0}^{n\times n}, \mathbb{R}_{>0}^{n\times n},\mathbb{R}^{n\times n}$ to be a non-negative, positive and a general real-valued $n \times n$ matrix, respectively. A vector of size $N$ with all ones is denoted by $1_N$. The identity matrix of dimension $N$ is denoted by $I_N$. The small-signal value of a variable $x$ is denoted by $\Delta x = x - x^{\rm ss}$, the deviation of $x$ (linearized) from its steady-state value $x^{\rm ss}$. For any complex number $z\in\mathbb{C}$, $\mathbb{R}\{z\},\mathbb{I}\{z\}$ denotes the real and imaginary parts of $z$, respectively.  We let the quantity $[N]$ denote the set $\{1,2,\ldots,N\}$. Within $[N]$, the set of all numbers excluding an entry $i$ (i.e., $\{1,\ldots, i-1,i+1, \ldots,N\}$ ) is given by $[N]\backslash \{i\}$. We denote the quantity $z^\star$ to be the complex conjugate of $z$. We generate Gaussian distributed scalars (vectors) using $\mathcal{N}(a,b)$, where, $a$ represents the mean (mean vector) and $b$ represents the variance (covariance matrix). The expectation operator is denoted by $\mathbb{E}[\cdot]$. We denote $\lVert x \rVert$ to be the Euclidean norm and $\lVert x \rVert_S^2 :=x^\top S x$. We denote the $p^{th}$ entry of a vector $x$ as $[x]_p$. The quantities $\mathrm{Mean} \{ x(k) \}$ and $\mathrm{Var} \{ x(k) \}$ refer to the mean and variance of $x(k)$ across data points $k=0,1,2,\ldots N$, respectively.

\section{Problem Setup}\label{problemsetupsection}
In this section, we discuss the grid-connected converter model, the small-signal network dynamics, grid-equivalent model and present our estimation problem.

\subsection{Grid-connected Converter Model}
\begin{figure}
    \centering
    \includegraphics[scale=0.5]{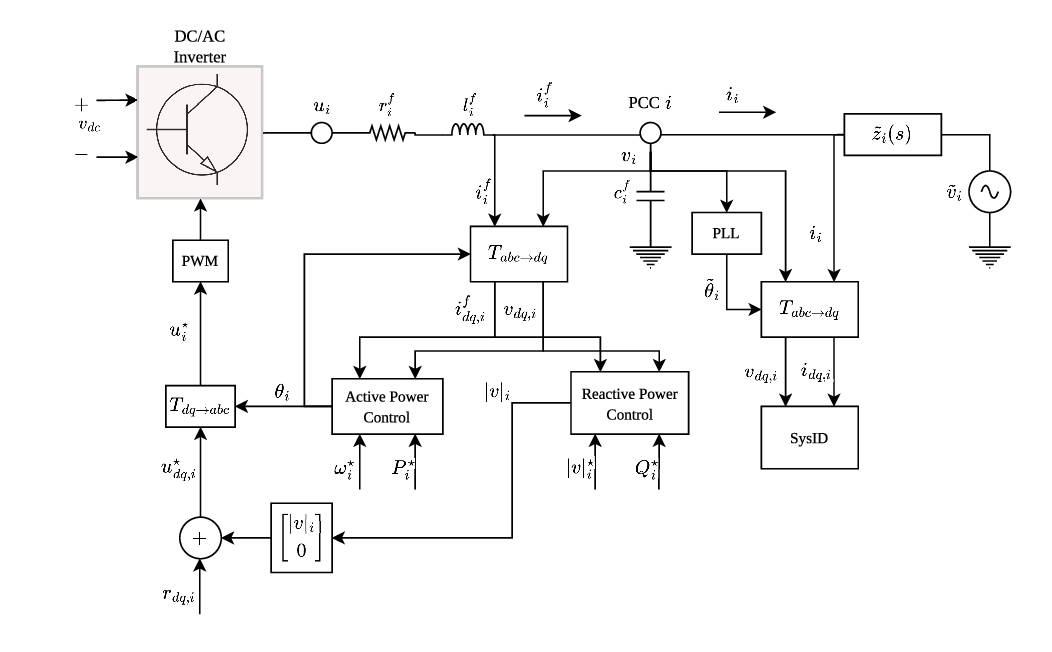}
    \caption{Grid-forming control scheme considered in our work \cite{directvoltageGFM}. Note that the equivalent voltage $\Tilde{v}_i$ is not considered to be a rigid infinite bus.}
    \label{fig:GFM}
\end{figure}

Fig. \ref{fig:GFM} shows the three-phase converter control design considered in our work. The converters are developed in grid-forming (GFM) mode \cite{directvoltageGFM} with droop control, although other GFM control strategies can also be used without any loss of generality. 

A DC-AC inverter transforms power from the DC source with voltage $v_{dc,i}$ into a three-phase AC sinusoidal voltage $u_{i}$ with the help of pulse-width modulation that dictates the frequency of switching of the inverter. We use an LC-filter at the converter side to smooth the sinusoids. The filter dynamics (in three-phase) are given by:
\begin{align}
        \frac{d}{dt}i_{i}^f &= \frac{1}{l_i^f}u_i - \frac{1}{l_i^f}v_i - \frac{r_i^f}{l_i^f}i_i^f, \  \frac{d}{dt}v_{i} = \frac{1}{c_i^f}(i_i^f - i_{i}), \label{filterdynamics}
\end{align}
where $i_i^f,i_i$ are the currents injected through the filter and at PCC $i$, respectively; $u_i,v_i$ are the voltages at the output of the DC-AC inverter and at PCC $i$, respectively. The filter inductance, resistance, and capacitance are given by $l_i^f,r_i^f,c_i^f$, respectively.

The filter current $i_i^f$ and PCC voltage $v_i$ are then converted to their synchronous rotating reference frame ($dq$ coordinates) using the Park transform \cite[Chapter 3.3]{kundur}
\begin{equation}\label{abctodq}
    i_{dq,i}^f = T_{abc\rightarrow dq0}(\theta_i)i_i^f, v_{dq,i} = T_{abc\rightarrow dq0}(\theta_i)v_i,
\end{equation}
where $\theta_i$ is the instantaneous voltage angle at the converter side. The instantaneous active and reactive powers $P_i = \frac{3}{2}v_{dq,i}^\top i_{dq,i}^f$, $Q_i = \frac{3}{2}v_{dq,i}^\top J i_{dq,i}^f$, are computed, respectively, with $J= \begin{bmatrix}
            0 & -1\\1 & 0
\end{bmatrix}$. The instantaneous active and reactive powers are then passed through low-pass filters:
\begin{equation*}
    \frac{d}{dt}P_{i,filt} = \omega_i^c(P_i - P_{i,filt}) \ \& \
        \frac{d}{dt}Q_{i,filt} = \omega_i^c(Q_i - Q_{i,filt}),
\end{equation*}
where, $\omega_i^c$ is the cut-off frequency for droop control. The droop control laws for voltage frequency and magnitude are then given by:
\begin{align}
    \omega_i = \omega_i^r - k_i^\omega(P_{i,filt}-P_{i}^r),
    |v|_i = |v|_i^r - k_i^v(Q_{i,filt}-Q_{i}^r) \label{droop},
\end{align}
respectively, where, $\omega_i,|v|_i$ are the real-time voltage frequency and magnitude, respectively, with $\omega_i^r,|v|_i^r$ being their desired references, and $k_i^\omega,k_i^v$ the tunable droop control parameters. Further, $P_i^r,Q_i^r$ are the desired active and reactive power references, respectively. We then compute the angle at the converter side using the relation $\frac{d}{dt}\theta_i = \omega_i$, which is used in the $abc$ to $dq$ transformation matrix in \eqref{abctodq}.

The pulse-width modulation (PWM) signal in the $dq$ coordinates is given by $u_{dq,i}^* = \begin{bmatrix}
    |v|_i\\
    0
\end{bmatrix} + r_{dq,i}$, where, $r_{dq,i}$ is an excitation signal for the purpose of estimation. We approximate the dynamics of the DC-AC inverter using the "average-value model", thus stating $u_i \approx T_{dq0\rightarrow abc}(\theta_i)[(u_{dq,i}^*)^\top \ 0]^\top$, with $T_{dq0\rightarrow abc}(\theta_i) = T_{abc\rightarrow dq0}(\theta_i)^{-1}$.

At each PCC, we record the current injected and the voltage in the $dq$ coordinates, i.e., $i_{dq,i},v_{dq,i}$, respectively. The conversion from $abc$ to $dq$ coordinates is carried out using the grid-angle measured using a phase-locked loop (PLL).

\subsection{Small-signal Network Model}
In this work, we consider a network of multiple three-phase voltage source converters (VSC) (as nodes of the network) interconnected with one another via transmission lines (as edges of the network) of non-zero impedance, as shown in Fig. \ref{fig:nconvertersystem}. The small-signal dynamics between PCC $i$ and $j$ in the complex coordinates (i.e., $\Delta\mathbf{v}_i := \Delta v_{d,i} + \mathbf{j}\Delta v_{q,i}$, $\Delta\mathbf{i}_i := \Delta i_{d,i} + \mathbf{j}\Delta i_{q,i}$, where, $x_{d,i}, x_{q,i}, \text{with } x\in \{ \Delta v,\Delta i\}$ represent the small-signal variables in the $d,q$ coordinates, respectively) \cite{decentralizedstability} is given by $\Delta \mathbf{i}_{ij} = \mathbf{Y}_{ij}(s)(\Delta \mathbf{v}_{i}-\Delta \mathbf{v}_{j})$, where, $\Delta \mathbf{i}_{ij}$ is the small-signal current from converter $i$ to $j$, $\mathbf{Y}_{ij}(s)=\mathbf{Z}_{ij}(s)^{-1} \in \mathbb{C}$ is the admittance (inverse of impedance) of the transmission line between $i,j$ and $\Delta \mathbf{v}_{i},\Delta \mathbf{v}_{j}$ are the small-signal voltages at the (PCC) generated by converters $i,j$, respectively. The current injection at each PCC $i$ in a network of $N$ converters is given by $\Delta \mathbf{i}_{i} = \sum_{j=1,j\neq i}^N \Delta \mathbf{i}_{ij}$, leading to the dynamic grid model $\Delta\mathbf{i} = \mathbf{Y}(s) \Delta \mathbf{v}$:
\begin{equation}\label{gridmodel}
\underset{\Delta \mathbf{i}}{\underbrace{\begin{bmatrix}
            \Delta \mathbf{i}_{1}\\
            \Delta \mathbf{i}_{2}\\
            \vdots\\
            \Delta \mathbf{i}_{N}
        \end{bmatrix}}}
         = \underset{\mathbf{Y}(s)}{\underbrace{\begin{bmatrix}
            \mathbf{Y}_{11}(s) & -\mathbf{Y}_{12}(s) & \ldots & -\mathbf{Y}_{1N}(s)\\
            -\mathbf{Y}_{21}(s) & \mathbf{Y}_{22}(s) & \ldots & -\mathbf{Y}_{2N}(s)\\
            \vdots & \vdots & \ddots & \vdots\\
            -\mathbf{Y}_{N1}(s) & -\mathbf{Y}_{N2}(s) & \ldots & \mathbf{Y}_{NN}(s) 
        \end{bmatrix}}}
        \underset{\Delta\mathbf{v}}{\underbrace{\begin{bmatrix}
            \Delta \mathbf{v}_{1}\\
            \Delta \mathbf{v}_{2}\\
            \vdots\\
            \Delta \mathbf{v}_{N}
        \end{bmatrix}}}
        ,
    \end{equation}
where, $\mathbf{Y}_{ii}(s) = \sum_{j=1,j\neq i}^N \mathbf{Y}_{ij}(s), \ \forall \ i \in [N]$.

\subsection{Grid-equivalent Model}
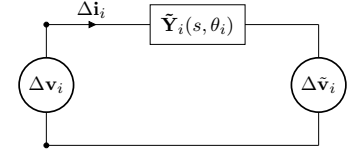
\begin{figure}[h]
    \centering
    \scalebox{0.8}{\begin{circuitikz}[>=Stealth]
    \ctikzset{bipoles/length=1cm}
    \ctikzset{sources/scale=1.5}
    
    \draw (0,0) to[esource, *-*] (0,2);
    \node at (0,1) {$\Delta \mathbf{v}_i$};

    \draw (0,2) to[short, i=$\Delta\mathbf{i}_i$] (1.5,2) -- (4.5,2);
    \node[draw, fill=white, inner xsep=5pt, inner ysep=4pt] at (2.5,2) {$\mathbf{\Tilde{Y}}_i(s, \theta_i)$};
          
    \draw (4.5,2) to[esource] (4.5,0);
    \node at (4.5,1) {$\Delta \Tilde{\mathbf{v}}_{i}$};
    
    \draw (4.5,0) -- (0,0);
    
    \end{circuitikz}}
    \caption{Thevenin equivalent of Fig. \ref{fig:nconvertersystem} from the perspective of PCC $i$.}
    \label{fig:thevenin}
\end{figure}
To derive a grid-equivalent model from the perspective of each VSC, we consider the Thevenin equivalent from VSC $i$ of the grid, as seen in Fig. \ref{fig:thevenin}. We now partition the grid model in \eqref{gridmodel} into two components: one corresponding to the local PCC $i$ and the rest of the grid, denoted by $-i$:
\begin{equation}\label{gridpartition}
    \begin{bmatrix}
            \Delta \mathbf{i}_i\\
            \Delta \mathbf{i}_{-i}
        \end{bmatrix}
         = \begin{bmatrix}
             \mathbf{Y}_{ii}(s) & \mathbf{Y}_{i,-i}(s)\\
             \mathbf{Y}_{i,-i}^\top(s) & \mathbf{Y}_{-i,-i}(s)
         \end{bmatrix} \begin{bmatrix}
            \Delta \mathbf{v}_i\\
            \Delta \mathbf{v}_{-i}
        \end{bmatrix},
\end{equation}
where, $\Delta\mathbf{i}_{-i}:= [\Delta \mathbf{i}]_j \in \mathbb{C}^{N-1 \times1}: j\in[N]\backslash\{i\}, \Delta\mathbf{v}_{-i}:= [\Delta \mathbf{v}]_j \in \mathbb{C}^{N-1 \times1}: j\in[N]\backslash\{i\} $ represents the vector of the small-signal current and voltage at all PCCs except that corresponding to $i$, $\mathbf{Y}_{i,-i} := -[\mathbf{Y}]_{ij} \in \mathbb{C}^{1 \times N-1}: j\in[N]\backslash\{i\}$ is a vector of the admittances linking PCC $i$ and the others in the grid, and $\mathbf{Y}_{-i,-i}(s) = [\mathbf{Y}]_{lj} \in \mathbb{C}^{N-1 \times N-1}: l,j \in [N]\backslash \{i\}$ represents the grid admittance matrix $\mathbf{Y}(s)$ with the row and column entries corresponding to PCC $i$ being removed.

The PCC voltages $\Delta \mathbf{v}_i$ are related to the small-signal internal voltages $\Delta \mathbf{u}_i$ (see Fig. \ref{fig:GFM}) with $\Delta\mathbf{v}_i = \Delta \mathbf{u}_i - \mathbf{Z}_i^f(s) \Delta \mathbf{i}_i$, where $\mathbf{Z}_i^f(s) := r_i^f + (s+\mathbf{j})l_i^f$ represents the filter impedance\footnote{We ignore the contribution of the capacitor in our theoretical analysis due to (i) being low in magnitude in traditional GFM converters, (ii) its influence being prominent only at higher frequencies, which is not of interest in this work.}. With the filter dynamics included in \eqref{gridpartition} ({for more details, please refer to Appendix \ref{eqadmvoltcomp}}), we derive the following relation for the grid equivalent model: 
\begin{equation}\label{thevenin}
    \Delta \mathbf{i}_i(s) = \mathbf{\Tilde{Y}}_i(s)(\Delta\mathbf{v}_i(s) - \Delta \Tilde{\mathbf{v}}_i(s)),
\end{equation}
where, the equivalent admittance and voltage\footnote{Henceforth, we drop the argument $(s)$ in the equivalent admittance and voltage due to space constraints} are computed as:

\footnotesize
\begin{align}
    \mathbf{\Tilde{Y}}_i &= \mathbf{Y}_{ii} - \mathbf{Y}_{i,-i}\mathbf{Z}_{-i}^f(I_{N-1} + Y_{-i,-i}\mathbf{Z}_{-i}^f)^{-1}\mathbf{Y}_{i,-i}^\top, \label{equivalentadmittance_full}\\
    \Delta \Tilde{\mathbf{v}}_i & = -\frac{\mathbf{Y}_{i,-i}}{\mathbf{\Tilde{Y}}_i}( I_{N-1} - \mathbf{Z}_{-i}^f(I_{N-1} + Y_{-i,-i}\mathbf{Z}_{-i}^f)^{-1}\mathbf{Y}_{-i,-i})\Delta \mathbf{u}_{-i}, \label{equivalentvoltage_full}
\end{align}
\normalsize respectively, with ${\mathbf{Z}_{-i}^f} := \mathrm{diag}[\mathbf{Z}_1^f, \ldots, \mathbf{Z}_{i-1}^f, \mathbf{Z}_{i+1}^f, \ldots \mathbf{Z}_N^f] \in \mathbb{C}^{N-1\times N-1}$ being a diagonal matrix comprising the filter impedance of all the VSCs except that of $i$ and $\Delta\mathbf{u}_{-i}:= [\Delta \mathbf{u}]_j \in \mathbb{C}^{N-1 \times1}: j\in[N]\backslash\{i\}$ representing the vector of internal voltages of all VSCs except that of $i$.
\begin{remark}[Complex coordinates]\label{complexcoordinatesremark}
    Note that if the active (i.e., the equivalent voltage $\Delta \Tilde{\mathbf{v}}_i$) and passive (i.e., the equivalent admittance $\mathbf{Y}_i$) entities were not distinguished, and a parametric approach was made to fit the overall equivalent model like in \cite{mimoidentificationverena,ultrafast}, we would not have an SISO setup like in \eqref{thevenin}, rather a $1\times 2$ MIMO transfer function to be estimated \cite{ultrafast}, owing to the $dq$-asymmetry arising from the active components.
\end{remark}
\subsection{Problem Statement}
Prior to stating the formal problem statement, we make the following assumptions on the line and filter impedances:
\begin{assumption}[Line \& filter impedances]\label{RLassumption}
    The interconnecting lines and the filters have the following properties:
    \begin{enumerate}
        \item[(i)] All interconnecting lines are resistive-inductive, with $r_{ij},l_{ij}$ (in p.u.) representing the resistance and inductances between PCC $i$ and $j$, respectively.
        \item[(ii)] The lines are homogeneous, i.e. $\rho_{ij} = r_{ij}/l_{ij} = \rho << 1, \ \forall \ i,j\in [N]$.
        \item[(iii)] The filters associated with each VSC are also homogeneous, i.e., $\rho_i^f = r_i^f/l_i^f = \rho^f << 1, \ \forall \ i\in [N]$.
    \end{enumerate}
\end{assumption}
Assumption \ref{RLassumption}(i) is commonly used in the power systems literature \cite{mininvasiveEKF,decentralizedstability,communicationfree}. Assumptions \ref{RLassumption}(ii),(iii) are used merely to fix a parametric structure on the impedance and to be able to provide theoretical guarantees, and are not a hard requirement\footnote{Nevertheless, the transmission lines and converter filters tend to be predominantly inductive, as evidenced in \cite{decentralizedstability}}. Our algorithm works efficiently for heterogeneous networks as well (refer to Section \ref{resultssection}).

With the above assumptions, we obtain the following simplified relations for the individual entries in \eqref{equivalentadmittance_full}, \eqref{equivalentvoltage_full}:
\begin{align}\label{simplifiedeqns}
    \mathbf{Y}_{ii} = \frac{1}{s + \mathbf{j} + \rho}\sum_{m=1,m\neq i}^N \gamma_{im} \ & , \ \mathbf{Y}_{i,-i} = \frac{1}{s +\mathbf{j} + \rho}\Gamma_{i,-i}, \\
    \mathbf{Y}_{-i,-i} =  \frac{1}{s + \mathbf{j} + \rho} \Gamma_{-i,-i} \ & , \ \mathbf{Z}_i^f = ({s + \mathbf{j} + \rho^f})l_i^f, \nonumber
\end{align}
where, $\gamma_{im} = 1/l_{im}$ is the inverse of the inductance between PCC $i$ and $m$. The matrix comprising of $\gamma_{im}$ is given by the Laplacian matrix $\Gamma$ with its diagonal entries being $\gamma_{ii} = \sum_{m=1,m\neq i}^N \gamma_{im} \ , \ i \in [N]$. This leads to $\Gamma_{i,-i} = [\Gamma]_{im}\in \mathbb{R}_{>0}^{1\times N-1}: m \in[N]\backslash \{i\}$, the vector of the inverse of inductance values linking PCC $i$ and the others in the grid and $\Gamma_{-i,-i} = [\Gamma]_{lj} \in \mathbb{R}_{>0}^{N-1 \times N-1}: l,j \in [N]\backslash \{i\}$ representing the matrix $\Gamma$ with the row and column entries corresponding to PCC $i$ being removed. Finally, with  ${L_{-i}^f} := \mathrm{diag}[l_1^f, \ldots, l_{i-1}^f, l_{i+1}^f, \ldots, l_N^f] \in \mathbb{R}_{>0}^{N-1\times N-1}$ representing the diagonal matrix of the inductance of all the filters with the exception of VSC $i$, we have the following simplified relation for the equivalent admittance and voltage (for more details on the simplification, please refer to Appendix \ref{eqadmvoltcompsimple}), respectively:
\begin{equation}\label{equivalentmodelsimple}
    \Tilde{\mathbf{Y}}_i = \frac{\Tilde{\gamma}_i}{s + \mathbf{j} + \rho} \ , \ \Delta \Tilde{\mathbf{v}}_i = - \frac{1}{\Tilde{\gamma}_i}\Tilde{\Gamma}_{i,-i}\Delta \mathbf{u}_{-i},
\end{equation}
where, $\Tilde{\gamma}_i = \gamma_{ii} - \Gamma_{i,-i}L_{-i}^f(I_{N-1} + \Gamma_{-i,-i}L_{-i}^f)^{-1}\Gamma_{i,-i}^\top$ and $\Tilde{\Gamma}_{i,-i} = \Gamma_{i,-i}(I_{N-1} - L_{-i}^f(I_{N-1} + \Gamma_{-i,-i}L_{-i}^f)^{-1}\Gamma_{-i,-i})$. We now note the following lemma on the equivalent admittance and voltage for heterogeneous transmission lines and filters that are not necessarily predominantly inductive:
\begin{lemma}[Robustness of equivalent model to heterogeneous lines]\label{heterolines}
    Let Assumption \ref{RLassumption}(i) alone hold true. Consider the transmission lines and the converter filters to be heterogeneous and not necessarily predominantly inductive. The relative order of the equivalent admittance in \eqref{equivalentadmittance_full} is unity and that of the transfer function vector $\frac{\mathbf{Y}_{i,-i}}{\mathbf{\Tilde{Y}}_i}( I_{N-1} - \mathbf{Z}_{-i}^f(I_{N-1} + Y_{-i,-i}\mathbf{Z}_{-i}^f)^{-1}\mathbf{Y}_{-i,-i})$ in \eqref{equivalentvoltage_full} is zero.
\end{lemma}
\begin{proof}
    Refer to Appendix \ref{heteroproof}.
\end{proof}
The above lemma shows that even if Assumptions(ii),(iii) did not hold true, the relative degree of $\Tilde{\mathbf{Y}}_i$ is always one, exhibiting a similar Bode characteristic to the first-order approximation in \eqref{equivalentmodelsimple}, thus justifying the first-order parameteric setting for the equivalent admittance.

From the perspective of PCC $i$, it is desired to obtain a reliable representation of the entire equivalent grid, with only local information (i.e. $\Delta \mathbf{i}_i,\Delta \mathbf{v}_i$). The model parameters $\theta_i := \begin{bmatrix}
    \rho & \Tilde{\gamma}_i
\end{bmatrix}^\top$ and the small-signal equivalent grid voltage $\Delta \Tilde{\mathbf{v}}_{i}$ (across frequencies), are unknown. This leads to the following formal problem statement:
\begin{problem}\label{mainproblemstatement}
    Design a minimally invasive, decentralized gray-box identification algorithm that estimates $\theta_i, \Tilde{\mathbf{v}}_{i} \ \forall \ i\in[N]$ from each converter without interrupting grid operations, withstanding deviations from nominal conditions. Moreover, the algorithm should be able to be carried out in parallel across all converters.
\end{problem}
Before describing the solution to the above problem, we make the following justification regarding our gray-box approach (i.e., a parametric and non-parametric representation for the equivalent admittance and voltage, respectively):
\begin{remark}[Gray-box approach]\label{grayboxremark}
     The equivalent admittance $\Tilde{\mathbf{Y}}_i(\mathbf{j}\omega)$ represents an aggregated representation of the passive transmission lines connected to PCC $i$ with a well-defined, low-relative-order mathematical structure, making it convenient for parametric estimation. On the other hand, the equivalent voltage $\Tilde{\mathbf{v}}_{i}(\mathbf{j}\omega)$ captures an aggregated representation of the closed-loop reactions of all other VSCs in the grid. Since the exact control techniques (for example, the type of GFM, grid-following control algorithms) or the dynamical structure of synchronous machines and loads are completely unknown from the perspective of PCC $i$, enforcing a parametric model on $\Tilde{\mathbf{v}}_{i}(\mathbf{j}\omega)$ would not be scalable. By estimating the equivalent voltage in a non-parametric manner, the local VSC can accurately map the grid's frequency-dependent stiffness without making restrictive assumptions about the grid. In addition, assuming an overall parametric model for the net equivalent impedance as in \cite{mimoidentificationverena,ultrafast} would not be applicable here since the grid is not assumed to be rigid (i.e., an infinite bus) and is realistically composed of other low-inertia sources that would excite simultaneously.
\end{remark}

\section{Decentralized Identification Algorithm}\label{algorithmsection}
In this section, we propose a frequency domain approach to solve Problem \ref{mainproblemstatement}. We discuss the measurements recorded, present the estimation framework, the frequency domain approach to estimation, and finally, present our algorithm.
\subsection{Measurements}\label{excitations_methods}
As seen in Fig. \ref{fig:GFM}, at PCC $i$, we measure $i_{dq,i},v_{dq,i}$, the local current and voltage, respectively in the $dq$ coordinates, transformed with the angle measured using a PLL at the PCC. We then compute their deviations from their steady-state value\footnote{In the stochastic setting with multiple VSCs exciting simultaneously, the temporal mean at steady-state is considered as the steady-state value} $\Delta i_{dq,i} = i_{dq,i} - i_{dq,i}^{ss}$ and $\Delta v_{dq,i} = v_{dq,i} - v_{dq,i}^{ss}$, respectively. The signals are then converted to their complex coordinate counterparts, i.e. $\Delta \mathbf{i}_i := \Delta i_{d,i} + \mathbf{j}\Delta i_{q,i}$, $\Delta \mathbf{v}_i := \Delta v_{d,i} + \mathbf{j}\Delta v_{q,i}$. We then sample these signals with a sampling frequency $f_s$ and compute their fast Fourier transform (FFT) values\footnote{Note that although the estimation framework is in the continuous domain, we use the FFT as an approximation for the continuous-time Fourier transform, since the frequency response in the continuous domain and its sampled discrete-time counterpart are identical for frequencies significantly lower than the Nyquist frequency.} $\Delta \mathbf{i}_i(\mathbf{j}\omega_k),\Delta \mathbf{v}_i(\mathbf{j}\omega_k)$, at frequencies $\omega_k = 2\pi (f_s/N_s)k$, with $f_s$ being the sampling frequency of the signals $\Delta \mathbf{i}_i,\Delta \mathbf{v}_i$ and $N_s$ the total number of samples, and $k\in\{0,1,2\ldots,N_s\}$ the iteration count.

\subsection{Grid-Equivalent Model Estimation Framework}\label{estimationframework}
We compute the ratio $h_i(\mathbf{j}\omega_k) = \Delta \mathbf{i}_i(\mathbf{j}\omega_k)/\Delta \mathbf{v}_i(\mathbf{j}\omega_k)$ at each $\omega_k$. Dividing both sides of \eqref{thevenin} by $\mathbf{v}_i(\mathbf{j}\omega_k)$, we obtain the following map
\begin{equation}\label{thevinin1}
    h_i(\mathbf{j}\omega) = \Tilde{\mathbf{Y}}_i(\mathbf{j}\omega)(1-\Tilde{h}_i(\mathbf{j}\omega)),
\end{equation}
where, $\Tilde{\mathbf{Y}}_i(\mathbf{j}\omega) = \frac{\Tilde{\gamma}_i}{\rho + \mathbf{j}(\omega+1)}$ is the equivalent admittance and $\Tilde{h}_i(\mathbf{j}\omega) = \Delta \mathbf{\Tilde{v}}_i(\mathbf{j}\omega)/\Delta \mathbf{v}_i(\mathbf{j}\omega)$.

We then write $h_i,\Tilde{h}_i$ explicitly in terms of its real and imaginary components, i.e., $h_i(\mathbf{j}\omega) = \mathbb{R}\{h_i(\mathbf{j}\omega)\} + \mathbf{j}\mathbb{I}\{h_i(\mathbf{j}\omega)\}, \Tilde{h}_{i}(\mathbf{j}\omega) = \mathbb{R}\{ \Tilde{h}_i(\mathbf{j}\omega)\} + \mathbf{j}\mathbb{I}\{ \Tilde{h}_i(\mathbf{j}\omega)\}$, respectively. Re-arranging the terms in \eqref{thevinin1}, we obtain the following bilinear map

\footnotesize
\begin{equation}\label{maineqn}
    \underset{z_i(\omega)}{\underbrace{\begin{bmatrix}
            -(\omega+1)\mathbb{I}\{h_i(\mathbf{j}\omega)\}\\
            (\omega+1)\mathbb{Re}\{h_i(\mathbf{j}\omega)\}
        \end{bmatrix}}} 
        =  \underset{H_i(\omega)}{\underbrace{\begin{bmatrix}
            -\mathbb{Re}\{h_i(\mathbf{j}\omega)\} & 1\\
            -\mathbb{I}\{h_i(\mathbf{j}\omega)\} & 0
        \end{bmatrix}}}\underset{\theta_i}{\underbrace{\begin{bmatrix}
            \rho\\
            \Tilde{\gamma}_i
        \end{bmatrix}}} -\underset{[\theta_i]_2}{\underbrace{\Tilde{\gamma}_i}}\underset{d_i(\omega)}{\underbrace{\begin{bmatrix}
            \mathbb{Re}\{\Tilde{h}_i(\mathbf{j}\omega)\}\\
            \mathbb{I}\{\Tilde{h}_i(\mathbf{j}\omega)\}
        \end{bmatrix}}},
\end{equation}
\normalsize where, $\theta_i\in \mathbb{R}_{\geq 0}^{2\times 1}$ is the parameter vector corresponding to $\Tilde{\mathbf{Y}}_i(\mathbf{j}\omega)$ and is constant, while $d_i(\omega) \in \mathbb{R}^{2\times 1}$ represents the real and imaginary parts of $\Tilde{h}_i$, and varies across frequencies.

With the model in \eqref{maineqn}, it is desired to estimate $\theta_i$ and $d_i(\omega_k), \ \forall \omega_k\in\{0,2\pi f_s/N_s,4\pi f_s/N_s,\ldots \pi f_s/N_s\}$, for all converters $i\in[N]$ in a parallel and decentralized manner, leading to the following optimization problem

\scriptsize
\begin{equation}\label{mainopt}
    \underset{[\theta_i,\{d_i(\omega_k)\}_{k=0}^{N_s}]}{\arg}\min \sum_{k=1}^{N_s} \Bigg\lVert z_i(\omega_k) - \begin{bmatrix}
        H_i(\omega_k) & [\theta_i]_2
    \end{bmatrix}\begin{bmatrix}
        \theta_i\\
        d_i(\omega_k)
    \end{bmatrix}\Bigg\rVert
    ^2_{R(\omega_k)^{-1}},
\end{equation}
\normalsize where, $R(\omega_k)$ is the measurement noise covariance of $z_i(\omega_k)$, and $\lVert x\rVert_S^2 := x^\top Sx$.
\begin{remark}[Non-convexity]\label{mainoptremark}
    Since the estimation model \eqref{maineqn} is bilinear in terms of $\theta_i,d_i$, this results in the above optimization problem being non-convex. In addition, $d_i$ varies across frequencies. Although the above problem has a feasible solution, a simultaneous estimation algorithm using gradient descent or recursive least-squares across frequency data points would yield sub-optimal results.
\end{remark}
Thus, we carry out the estimation of the overall equivalent model in two steps: (a) Parameter estimation by disregarding frequency points where the coupling $d_i(\omega)$ is significant, thereby converting \eqref{mainopt} to a linear regression problem, after which, (b) Equivalent voltage estimation is carried out.

\subsection{Parameter Estimation Algorithm}\label{parameterest}
For the sake of parameter estimation, it must be noted that the small-signal equivalent voltage $\Delta \Tilde{\mathbf{v}}_i$ (and hence, $d_i(\omega)$ in \eqref{maineqn}) cannot be treated as a mere disturbance/ exogenous input. This is owing to the dependency on the PCC voltages of all other VSCs, which in turn depend on the voltage at PCC $i$ through the injected current, thus exhibiting a non-trivial closed-loop behaviour. Transients at the local PCC actively trigger the control loops of neighboring converters, causing the equivalent voltage $\Delta \Tilde{\mathbf{v}}_i$ to be correlated with the local voltage $\Delta \mathbf{v}_i$. Standard least-squares estimation towards identifying $\Tilde{\mathbf{Y}}_i$ relies on the assumption that the input $\Delta \mathbf{v}_i$ is statistically independent of the "noise" ($\Delta \Tilde{\mathbf{v}}_i$ in this case), i.e., exogenous. However, in our closed-loop setup, it is not the case: applying direct identification methods, such as ordinary least-squares or a raw empirical transfer function estimate (ETFE) results in severe asymptotic bias (see Fig. \ref{fig:ETFE} in Section \ref{resultssection} for empirical evidence). The estimator absorbs the active control dynamics of the neighboring VSCs into the local admittance estimate, rendering the extraction of $\Tilde{\mathbf{Y}}_i$ highly inaccurate.

To minimize the closed-loop correlations, we employ instrumental variables (IV) in the frequency domain \cite{IVfreq} by selecting the local wide-band excitation signal in the complex coordinates $\mathbf{r}_i = r_{d,i} + \mathbf{j}r_{q,i}$ as the instrument. We now make an important assumption on the wide-band excitation signals of the VSCs.
\begin{assumption}[Orthogonal excitations]\label{excitationsassumption}
    The wide-band excitations of all the VSCs are mutually orthogonal, i.e., $\mathbb{E}_{t\in\{0,1,2\ldots,N_s\}}[\mathbf{r}_i^\star(t)\mathbf{r}_j(t)]\approx 0, \ \forall \ i,j \in [N], i\neq j$.
\end{assumption}
Although the above assumption implies statistical independence of all the exciting signals $\mathbf{r}_i$, it does not require centralized communication. On the contrary, in order for the exciting signals to be correlated, one requires centralized coordination. In a large-scale setup, since each VSC independently generates excitation signals, it is highly likely that the excitations are almost always out of phase, especially taking delays, converter and line harmonics into consideration.

Pre-multiplying $\mathbf{r}_i(s)$ on both sides of \eqref{thevenin} and taking the expectation, we obtain the following relation:
\begin{equation}\label{thevenin_IV}
    S_{i}^{RI}(\mathbf{j}\omega) = \Tilde{\mathbf{Y}}_i(\mathbf{j}\omega)( S_{i}^{RV}(\mathbf{j}\omega) -  S_{i}^{R\Tilde{V}}(\mathbf{j}\omega)),
\end{equation}
where, $S_{i}^{RI}(\mathbf{j}\omega) = \mathbb{E}[\mathbf{r}_i^\star(\mathbf{j}\omega)\Delta \mathbf{i}_i(\mathbf{j}\omega)]$, $S_{i}^{RV}(\mathbf{j}\omega) = \mathbb{E}[\mathbf{r}_i^\star(\mathbf{j}\omega)\Delta \mathbf{v}_i(\mathbf{j}\omega)]$ and $S_{i}^{R\Tilde{V}}(\mathbf{j}\omega) = \mathbb{E}[\mathbf{r}_i^\star(\mathbf{j}\omega)\Delta \Tilde{\mathbf{v}}_i(\mathbf{j}\omega)]$ represent the cross-spectral densities of the local wide-band excitation $\mathbf{r}_i$ with the small-signal current $\Delta \mathbf{i}_i$ and voltage $\Delta \mathbf{v}_i$ at PCC $i$, respectively, with $\mathbf{r}_i(\mathbf{j}\omega)$ representing the FFT of $\mathbf{r}_i(t)$. By projecting the instrument $\mathbf{r}_i$ onto the measured current and voltage via the cross-spectral density, we remove all uncorrelated noise. While $\mathbf{r}_i$ still remains correlated with the grid voltage $\Tilde{\mathbf{v}}_i$ due to the closed-loop effects (i.e., $\mathbb{E}[\mathbf{r}_i^\star\Tilde{\mathbf{v}}_i]\neq 0$), this pre-processing improves the possibility for the estimation to be separated between the passive admittance $\Tilde{\mathbf{Y}}_i(\mathbf{j}\omega)$ and the active contribution of the grid $\Tilde{\mathbf{v}}_i$.

We compute the spectral densities, and hence, the ratio $h_i'(\mathbf{j}\omega_k) = S_i^{RI}(\mathbf{j}\omega_k)/S_i^{RV}(\mathbf{j}\omega_k)$ at each $\omega_k$\footnote{Henceforth, we replace $\omega_k$ with $k$ for notational convenience. Thus, any term $x(\mathbf{j}\omega_k)$ is to be replaced with $x(k)$.}. We then solve the following least-squares problem:
\begin{equation}\label{theta_opt}
    \hat{\theta}_i =\underset{\theta_i}{\arg}\min \sum_{k=0}^{N_s} \Bigg\lVert z_i'(k) - H_i'(k)\theta_i\Bigg\rVert_{R_\theta(k)^{-1}}
    ^2,
\end{equation}
where, $z_i'(k) := \begin{bmatrix}
            -(\omega_k+1)\mathbb{I}\{h_i'(\mathbf{j}\omega_k)\}\\
            (\omega_k+1)\mathbb{Re}\{h_i'(\mathbf{j}\omega_k)\}
        \end{bmatrix}$ and $H_i'(k) := \begin{bmatrix}
            -\mathbb{Re}\{h_i'(\mathbf{j}\omega_k)\} & 1\\
            -\mathbb{I}\{h_i'(\mathbf{j}\omega_k)\} & 0
\end{bmatrix}$ represents the output and $H_i(\omega)$ of our measurement model framework in \eqref{maineqn}, but with $h_i$ being replaced by $h_i'$, $R_\theta(k) = \sigma_\theta(k)I_2$ is the measurement noise covariance associated with $z_i'(k)$. We set the measurement noise variance $\sigma_\theta(k):= c_1 + c_2(1-\mathcal{W}_i(k))$, with $c_2>>c_1$. By setting $\mathcal{W}_i(k) = 0$, this informs the unreliability of frequency point $\omega_k$ towards parameter estimation. In this manner, we avoid the errors associated with frequency data points where the closed-loop effects from the rest of the grid appear to cloud data for parameter estimation.

We now explain the frequency discrimination (FD) criteria (using the pseudo-code in Algorithm \ref{alg:FD}) to identify these unreliable frequency points (i.e., when $\mathcal{W}_i(k)=0$) for parameter estimation.
    \begin{enumerate}[label=(\Alph*)]
    \item \textbf{Coherence}: We compute the coherence factor $C_i^{RV}(k)$ between the local excitation $\mathbf{r}_i$ and small-signal voltage $\Delta \mathbf{v}_i$, describing the linear dependence between them. If it is found to be lower than a given $\epsilon$, the voltage response at that frequency is predominantly driven by the grid rather than the local VSC, and that data point is discarded.
    \item \textbf{Band-pass filtering}: To prevent estimation bias from known physical operating conditions, we discard frequencies outside a bounded range $[\omega^a,\omega^b]$. Frequencies below $\omega^a$ are rejected as they are heavily dominated by the fundamental grid frequency (DC value in the $dq$ and complex coordinates) and active closed-loop tracking dynamics (e.g., droop control, PLLs, etc.). Similarly, frequencies above $\omega^b$ are rejected to avoid spectral aliasing and high-frequency harmonics from the inverters.
    \item \textbf{Passivity awareness}: Note that $\mathbb{R}\{ \Tilde{\mathbf{Y}}_i(\mathbf{j}\omega)\}>0$ at all frequencies since $\Tilde{\gamma}_i>0,\rho\geq0$. Thus, if $\mathbb{R}\{h_i'(\mathbf{j}\omega)\}<0$ at any frequency $\omega$, it is certain that the noise ($\Tilde{h}_i$ in this case) completely overpowers the signal. We discard these frequencies for $\theta_i$ estimation.
\end{enumerate}
\begin{algorithm}
\caption{$\textrm{Frequency Discrimination}$} \label{alg:FD}
    \begin{algorithmic}
        \STATE $\mathcal{W}_i \leftarrow 1_{N_s}$ \textrm{Frequency weight}
        \FOR{$k=0$ \textit{to} $N_s$}
            \STATE $C_i^{RV}(k) = |S_i^{RV}(\mathbf{j}\omega_k)|^2/(S_i^{RR}(\mathbf{j}\omega_k)S_i^{VV}(\mathbf{j}\omega_k))$
            \IF[\textbf{(A) Coherence}]{$C_i^{RV}(k)<\epsilon$}  
                \STATE $\mathcal{W}_i(k) \leftarrow 0$
            \ENDIF
            \IF[\textbf{(B) Band-pass filtering}]{$\omega_k<\omega^a$ OR $\omega_k>\omega^b$}
            \STATE $\mathcal{W}_i(k) \leftarrow 0$
            \ENDIF
            \IF[\textbf{(C) Passivity awareness}]{$\mathbb{R}\{h_i'(k)\}<0$}
                \STATE $\mathcal{W}_i(k) \leftarrow 0$
            \ENDIF
        \ENDFOR
    \end{algorithmic}
\end{algorithm}

\subsection{Small-signal Equivalent Voltage Estimation Algorithm}\label{disturbanceest}
Upon estimating $\theta_i$ using \eqref{theta_opt}, we now implement a Kalman filter algorithm \cite{KFtheory} to estimate $d_i(k) \ \forall \ k\in \{0,1,\ldots, N_s\}$. The main idea is to treat this setup as an "unknown input observer" \cite{UIO}. To this end, we define the process model as a random walk
\begin{equation}\label{randomwalk}
    d_i(k+1) = d_i(k) + q_i(k), \ q_i(k) \sim \mathcal{N}(0_2,Q), 
\end{equation}
where, $d_i(k) := d_i(\omega_k)$ and process noise covariance $Q = \sigma_qI_2$.

The measurement model is given by $\Tilde{z}_i(k) = [\hat{\theta}_i]_2d_i(k) + \nu_i^d(k)$, with $\Tilde{z}_i(k) = {z}_i(k) - H_i^\theta(k)\hat{\theta}_i$ and the measurement noise $\nu_i^d(k) \sim \mathcal{N}(0_2,R_d)$, with the measurement noise covariance $R_d = c_1I_2$. We did not include process noise for the estimation of the parameters $\theta_i$ as they are fixed for linear time-invariant systems. On the other hand, $d_i$ varies across frequency; thus, it is necessary to include a process noise.
\begin{remark}[Process model]\label{KFhighgain}
    In reality, there exists no deterministic model to describe the evolution of $d_i$ along frequency data points. This intuitively prompts the process model to be $d_i(k) = q_i(k)$ rather than the random walk model in  \eqref{randomwalk}. However, given our hyper-parameter tuning, this leads to a high-gain observer, effectively resulting in a naive estimate $\hat{d}_i(k) \approx ([\hat{\theta}_i]_2)^{-1}\Tilde{z}_i(k)$, which is heavily sensitive to small deviations in $\hat{\theta}_i$ from its true value $\theta_i$. Moreover, when there is no recursion in the process model, the algorithm does not learn from the errors of previous frequency data points.
\end{remark}
Similar to \eqref{theta_opt}, we have the following optimization problem to solve for $d_i$:

\footnotesize
\begin{equation}\label{d_opt}
    \hat{d}_i(k) = \underset{d_i(k)}{\arg}\min  \Bigg\lVert \Tilde{z}_i(k) - [\hat{\theta}_i]_2d_i(k)\Bigg\rVert_{R_d(k)^{-1}}
    ^2 + \Bigg\lVert d_i(k) - \hat{d}_i(k-1)\Bigg\rVert_{Q^{-1}}^2,
\end{equation}
\normalsize  leading to the following posterior update equations:
\begin{align}
    \hat{d}_i(k+1) &= \hat{d}_i(k) + K_d(k)(\Tilde{z}_i(k) - [\hat{\theta}_i]_2\hat{d}_i(k)), \label{dupdate}\\
    P_d(k+1) &= (I_{2} - [\hat{\theta}_i]_2K_d(k))(P_d(k)+Q(k)), \label{dPupdate}
\end{align}
with the Kalman filter gain at frequency iteration $k$ given by 
\begin{equation}\label{KFgain}
    K_d(k) = [\hat{\theta}_i]_2(P_d(k)+Q)(R_d + [\hat{\theta}_i]_2^2(P_d(k)+Q))^{-1}.
\end{equation}

Once $d_i$ is estimated, it is then possible to ascertain the magnitude and phase of the small-signal equivalent voltage as $\Delta \hat{\Tilde{\mathbf{v}}}_i(\mathbf{j}\omega) = \Delta \mathbf{v}_i(\mathbf{j}\omega)([\hat{d}_i(\omega)]_1 + \mathbf{j}[\hat{d}_i(\omega)]_2)$.

\subsection{Theoretical Guarantees}
We now provide analytical guarantees on the error bounds associated with the equivalent admittance and voltage, and their inferences. Before providing the equivalent admittance and voltage error bounds, we provide the following remark on the number of solutions.

\begin{remark}[Number of solutions]\label{numberofsolutions}
Consider an error $\epsilon_i^\mathbf{Y}(\mathbf{j}\omega)$ on the estimated equivalent admittance such that $\hat{\Tilde{\mathbf{Y}}}_i(\mathbf{j}\omega) = {\Tilde{\mathbf{Y}}}_i(\mathbf{j}\omega) + \epsilon_i^\mathbf{Y}(\mathbf{j}\omega)$. This error gets algebraically absorbed in the equivalent voltage estimate, i.e., $\Delta \hat{\Tilde{\mathbf{v}}}_i(\mathbf{j}\omega) = \frac{\epsilon_i^\mathbf{Y}(\mathbf{j}\omega)}{\hat{\Tilde{\mathbf{Y}}}_i(\mathbf{j}\omega)}\Delta\mathbf{v}_i(\mathbf{j}\omega) + \frac{\Tilde{\mathbf{Y}}_i(\mathbf{j}\omega)}{\hat{\Tilde{\mathbf{Y}}}_i(\mathbf{j}\omega)}\Delta \Tilde{\mathbf{v}}_i(\mathbf{j}\omega)$ ensuring the overall equivalent model (i.e., $\Delta\mathbf{i}_i(\mathbf{j}\omega)/\Delta\mathbf{v}_i(\mathbf{j}\omega)$) remains unchanged. Consequently, there exists an infinite number of mathematical solutions for the split between equivalent admittance and voltage.

This flexibility often prompts approximations that appear convenient, such as assuming an infinite bus (i.e., $\Delta \Tilde{\mathbf{v}}_i = 0$) to lump all network complexities into the admittance (like in \cite{mimoidentificationverena,ultrafast}), or conversely, fixing a standard admittance $\Tilde{\mathbf{Y}}^0$ to shift all unmodeled dynamics into the voltage. Although the overall input-output model remains identical, these approaches are not recommended because: (a) the estimated variables lose their physical meaning, preventing the accurate extraction of critical internal parameters (e.g., inertia and damping constants) from the active components; and (b) forcing the equivalent admittance to absorb the active control dynamics of other VSCs negates the simple first-order structure that we have in our approach. Accounting for these active components within a passive admittance framework results in a more complex, asymmetric (in $dq$ coordinates) high-order transfer function, which complicates the parameter identification process and makes the algorithm highly susceptible to numerical ill-conditioning.
\end{remark}
We now present the formal results on the upper-bound on the parameter estimation error.
\begin{theorem}[Equivalent admittance estimation error]\label{theorem1}
Assume that there exists closed-loop transfer functions $G_{i}(s),\Tilde{G}_i(s), \ \forall \ i\in[N]$ in the Thevenin equivalent in Fig. \ref{fig:thevenin} with the individual excitations $\mathbf{r}_i$ as inputs and $\Delta \mathbf{v}_i,\Delta \Tilde{\mathbf{v}}_i$ as outputs, respectively, such that 
\begin{align}
    \Delta \mathbf{v}_i(\mathbf{j}\omega) &= G_i(\mathbf{j}\omega)\mathbf{r}_i(\mathbf{j}\omega) + \sum_{j=1,j\neq i}^{N} G_j(\mathbf{j}\omega)\mathbf{r}_j(\mathbf{j}\omega) + n_i(\mathbf{j}\omega), \label{closedloop1}\\
    \Delta\Tilde{\mathbf{v}}_i(\mathbf{j}\omega) &= \Tilde{G}_i(\mathbf{j}\omega)\mathbf{r}_i(\mathbf{j}\omega) + \sum_{j=1,j\neq i}^{N} \Tilde{G}_j(\mathbf{j}\omega)\mathbf{r}_j(\mathbf{j}\omega) + \Tilde{n}_i(\mathbf{j}\omega) \label{closedloop2},
\end{align}
where, $n_i(\mathbf{j}\omega),\Tilde{n}_i(\mathbf{j}\omega)$ are uncorrelated harmonic noise. Further, let Assumptions \ref{RLassumption},\ref{excitationsassumption} hold true. Then, the parameter estimation method in Section \ref{parameterest} renders the following upper-bound on the parameter estimation error $\epsilon_i^\theta := \hat{\theta}_i - \theta_i$:
\begin{equation*}
    \lVert \epsilon_i^\theta \rVert < \Tilde{\gamma}_i \mathcal{S}_i \sqrt{\frac{1 + \mathrm{Mean}\{ |h_i'(\mathbf{j}\omega_k)|^2\}}{\mathrm{Var}\{ \mathbb{R}\{ h_i'(\mathbf{j}\omega_k) \}\} + \mathrm{Mean}\{\mathbb{I}\{ h_i'(\mathbf{j}\omega_k) \}^2 \}}},
\end{equation*}
where, $\mathcal{S}_i = \sqrt{\mathrm{Mean}\{\vert\tilde{G}_i(\mathbf{j}\omega_k)/G_i(\mathbf{j}\omega_k)\vert^2\}}$.
\end{theorem}
\begin{proof}
    Refer to Appendix \ref{thm1proof}.
\end{proof}
We infer the following from the above theorem. (a) The parameter estimation error is directly proportional to $\Tilde{\gamma}_i$ (inverse of the equivalent grid inductance). When the transmission lines are highly inductive, $\Tilde{\gamma}_i$ is small. The grid's large inductance acts as a low-pass filter with a smaller cut-off frequency, thus damping higher-frequency noises that would otherwise lead to inaccurate parameter estimates. (b) Note that the denominator of the bound comprises the term $\mathrm{Var}\{ \mathbb{R}\{ h_i'(\mathbf{j}\omega_k) \}\}$. This term corresponds to the persistency of excitation. When we use a wide-band PRBS excitation, the variance is increased, thus leading to better estimates. (c) The term $\mathcal{S}_i$ represents the dynamic coupling (inversely related to grid stiffness). This quantifies the reaction of the rest of the grid to the local excitation. For a VSC connected to a strong grid (i.e., infinite bus), the grid does not react to local perturbations ($\tilde{G}_i \to 0$), driving $\mathcal{S}_i \to 0$ and hence the estimation error to zero. On the other hand, in low-inertia grids, the local excitation actively prompts other VSCs to react, leading to deviations in the parameter values unless an FD criterion like in Algorithm 1 is established. (d) The FD algorithm mathematically minimizes the error bound by constraining $\mathcal{S}_i$. Specifically, rejecting frequency points with low coherence prevents $G_i(\mathbf{j}\omega)$ from approaching zero; rejecting lower frequencies excludes the bandwidth where droop controllers actively operate (where $\tilde{G}_i(\mathbf{j}\omega)$ is exceptionally large); and implementing a passivity removing non-physical anomalies where $\vert{}\tilde{G}_i/G_i\vert{} \gg 1$. Together, these steps minimize $\mathcal{S}_i$.

We now present the formal results on the upper-bound on the equivalent voltage estimation error at each frequency $\omega_k$:
\begin{theorem}[Equivalent voltage estimation error]\label{theorem2}
Let Assumption \ref{RLassumption} hold true. With the parameter estimation error $\epsilon_i^\theta$ bounded by Theorem \ref{theorem1}, let $\epsilon_i^d(\omega_k) := \hat{d}_i(\omega_k) - d_i(\omega_k)$ be the error between the Kalman filter estimate and the true value of $d_i$ at frequency $\omega_k$.
The expected squared norm of the state estimation error is strictly bounded by the recursive relation:
\begin{align*}
    \mathbb{E}[\Vert{}\epsilon_i^d(\omega_k)\Vert{}^2] \le & \varphi^2 \Big( \mathbb{E}[\Vert{}\epsilon_i^d(\omega_{k-1})\Vert{}^2]
    +\mathbb{E}[\Vert{}\Delta d_i(\omega_k)\Vert{}^2] \Big) + \nonumber \\
     &\frac{2\lVert \epsilon_i^\theta\rVert^2}{\hat{\Tilde{\gamma}}_i^2}(1 + |h_i(\omega_k)|^2 + \mathbb{E} \Big[\lVert d_i(\omega_k) \big\Vert{}^2 \Big]),
\end{align*}
where, $\varphi \in (0,1)$ is a contraction constant and $\Delta d_i(\omega_k) := d_i(\omega_k) - d_i(\omega_{k-1})$ denotes the difference of the true value of $d$ between consecutive frequency data points. Since the estimation error for $d$ contracts along frequency points, so does the equivalent voltage estimation error given by $\mathbb{E}[\Vert{}e_i^v(\mathbf{j}\omega_k)\Vert{}^2] = \vert{}\Delta \mathbf{v}_i(\mathbf{j}\omega_k)\vert{}^2 \mathbb{E}[\Vert{}e_i^d(\omega_k)\Vert{}^2]$, where $e_i^v(\mathbf{j}\omega_k) := \Delta\hat{\tilde{\mathbf{v}}}_i(\mathbf{j}\omega_k) - \Delta\tilde{\mathbf{v}}_i(\mathbf{j}\omega_k)$.
\end{theorem}
\begin{proof}
    Refer to Appendix \ref{thm2proof}.
\end{proof}
We infer the following from the above theorem: (a) Similar to a time-domain Kalman filter, the estimation error asymptotically decays across frequency iterations. Because the low-frequency band contains the crucial active components (e.g., inertia and damping), we execute the filter backward across the spectrum (from $k=N_s$ to $k=0$). This ensures the estimation error has maximally converged precisely where accuracy is most critical. (b) The equivalent voltage estimation error is directly proportional to the squared parameter estimation error $\Vert{}\epsilon_i^\theta\Vert{}^2$ due to the explicit reliance on $\hat{\theta}_i$. Note that this error injection is amplified by a factor of $1/\hat{\gamma}_i^2$. Since $\hat{\gamma}_i$ is proportional to the inverse of the equivalent grid inductance, a highly inductive grid yields a correspondingly small $\hat{\gamma}_i$. Therefore, in highly inductive networks, the equivalent voltage estimation becomes very sensitive to parameter inaccuracies. This shows an interesting trade-off: while a large grid inductance acts as a natural low-pass filter that improves equivalent admittance estimation (as established in Theorem \ref{theorem1}), it simultaneously amplifies the propagation of any remaining parameter estimation errors into the equivalent voltage estimate.

\section{Numerical Case Studies}\label{resultssection}
In this section, we present the numerical values of the problem setup. Subsequently, we present the choice of hyperparameters, and then the simulation results of parameter and equivalent voltage estimation, and the main inferences. 

\subsection{Experimental Setup}
We consider an $N=5$ converter system interconnected with each other via a resistive-inductive line, all operating in GFM mode with the control structure as shown in Fig. \ref{fig:GFM}. The parameter values considered are shown in Table \ref{gridparameters}. The resistance and inductance values in p.u. between VSC $i,j$ are given by $r_{ij} = R_{ij}/Z_b,l_{ij} = L_{ij}/L_b$ ($R_{ij},L_{ij}$ being the resistance and inductance in S.I. values and the base values $i_b = S_b/\sqrt{3}v_b$, $Z_b = v_b/i_b$, $L_b = Z_b/\omega_b$) with their ranges as shown in Table \ref{gridparameters}, with resistance-to-inductance ratio $\rho_{ij} = r_{ij}/l_{ij}<<1$ \cite{decentralizedstability}. In addition, it is ensured that no two edges share the same $\rho_{ij}$, i.e., the lines are heterogeneous. 

We carry out $20$ Monte-Carlo experiments with our $5-$ converter system. Across the experiments, although the line and filter parameters remain constant, we consider different cut-off frequencies and droop control parameters in the VSCs, along with the excitation sequence varying across the experiments (see Table \ref{gridparameters} for parameter ranges). Pseudo-random binary sequences (PRBS) of magnitude $0.002$ p.u. are injected into the reference signal at each VSC (refer Fig. \ref{fig:GFM}), in line with the specified grid codes for harmonics: IEEE 519 \cite{519gridcode}. The initial phase angle for all converters is set to be zero.
\begin{table}[h]
    \begin{center}
    \caption{Parameter values of the grid \& Converters}\label{gridparameters}
    \begin{tabular}{p{4.3cm}|p{1cm}|p{2.4cm}}
    \hline
        Parameter & Symbol & Numerical value \\
        \hline
        Base power \& frequency & $S_b, \omega_b$ & $1.5\rm kVA, 100\pi \ rad/s$\\
        Base voltage \& current & $v_b, i_b$ & $380 \rm V, 2.28 A$\\
        Base impedance, inductance & $Z_b,L_b$ & $166.7\Omega, 0.53\rm H$\\
        Base capacitance & $C_b$ & $190\mu \rm F$\\
        Line resistance, R-L ratio (p.u.) & $r_{ij},\rho_{ij}$ & $[0.1,0.3],[0.05,0.1]$ \\
        Filter resistance, R-L ratio (p.u.) & $r_i^f,\rho_i^f$ & $[0.03,0.06],[0.2,0.4]$\\
        Filter capacitance (p.u.) & $c_i^f$ & $0.005$\\
        Cut-off frequency & $\omega_i^c$ & $2\pi [5,10] \ \rm rad/s$\\
        Droop control gains & $k_i^\omega, k_i^v$ & $[0.2,0.4]\times 1/S_b$\\
        Power setpoints (p.u.) & $P_i^r,Q_i^r$ & $ [0.9,1.1],0$\\
        Magnitude, frequency setpoints (p.u.) & $|v|_i^r,\omega_i^r$ & $[0.99,1.01],1$\\
        PRBS excitations (p.u.) & $r_{dq,i}$ & $\{ \pm 0.002 \}^{2 \times 1}$\\
        \hline
    \end{tabular}
    \end{center}
\end{table}
\begin{remark}[Simultaneous excitations]\label{excitationamplt}
    Note that for general SysID algorithms in the single converter case with an infinite bus, the larger the excitation magnitude (in this case, the PRBS $r_{dq,i}$), the better the estimation, for example, a magnitude of $0.1$ p.u. in \cite{mimoidentificationverena}. However, in our multi-agent setup, the grid voltage is not a rigid infinite bus. Thus, an increase in the excitation amplitude of one VSC, which would aid in the estimation of its equivalent model, would be detrimental to the estimation of other VSCs' equivalent impedance. Hence, it is vital that all VSCs strictly adhere to the grid code limitations on excitation amplitude.
\end{remark}
We then measure the current and voltage at PCC $i$ at a sampling rate $f_s = 10\rm kHz$, with the total simulation time being $N_{total} = 55\rm s$. 

\subsection{Hyperparameter Selection}
In the parameter estimation algorithm in Section \ref{parameterest}, we set the hyperparameters $c_1 = 0.1, c_2 = 10^{20}$, thus almost disregarding data points where the grid voltage corrupts data for parameter estimation. In the pseudo-code for frequency discrimination in Algorithm \ref{alg:FD}, we set $\epsilon=0.1$ to disregard frequency points predominantly dominated by the grid, the band-pass cut-off frequencies $\omega^a = 100 \rm \ rad/s$ and $\omega^b = 600 \rm \ rad/s$ to remove the biased data owing to the nominal grid operation and droop control in the lower frequencies, and harmonic noise in the higher frequencies, respectively.

\subsection{Simulation Results}\label{simulationresultssec}
In Fig. \ref{fig:ETFE}, we show the raw ETFE along with the true equivalent admittance $\mathbf{\Tilde{Y}}_i(s)$ from the perspective of VSC 1. Note that we use the formula in \eqref{equivalentadmittance_full} to compute the true equivalent admittance and not the simplified version as in \eqref{equivalentmodelsimple}. At low frequencies\footnote{We have employed our algorithm in the complex domain after converting to the $dq$ coordinates. Thus, low frequencies in the complex coordinates correspond to regions in the vicinity of the nominal operating frequency (i.e., $\omega_b$) in the $abc$ coordinates.}, the significant bias in the raw ETFE is due to the VSCs operating in the vicinity of the nominal frequency $\omega_b$ and the droop-control response of all VSCs to the PRBS excitation. We also note that there exists wide-band spectral noise at all frequencies owing to the simultaneous excitation of all VSCs, thus severely corrupting the data at hand. From Fig. \ref{fig:ETFE}, we see that in order to elegantly separate the estimation of $\mathbf{\Tilde{Y}}_i$ from the contribution of the active parts $\Delta \Tilde{\mathbf{v}}_i$, it is necessary to employ the frequency-domain pre-processing techniques described in Section \ref{parameterest}.

\begin{figure}[h]
    \centering
    \includegraphics[scale=0.38]{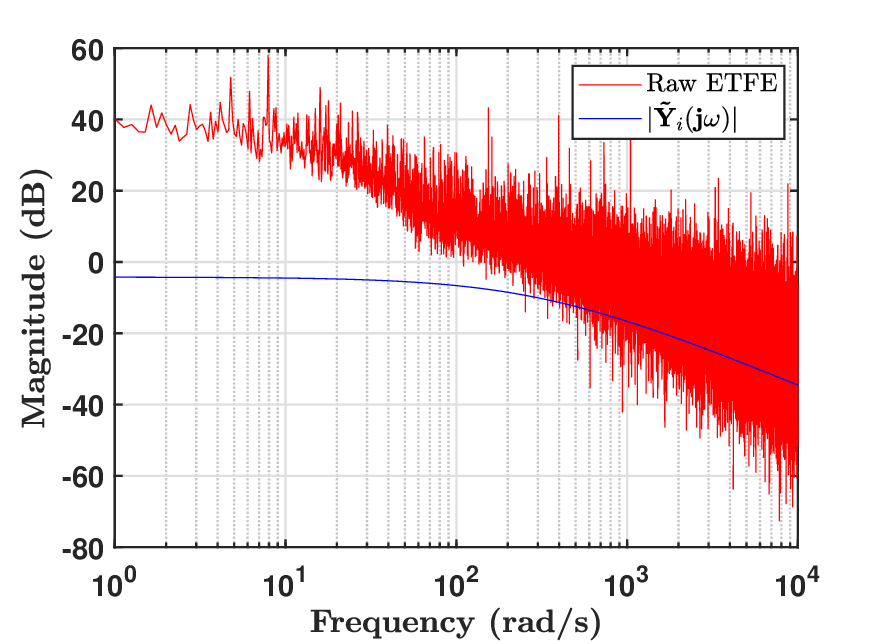}
    \caption{A representative plot of the true equivalent admittance as seen from VSC 1 (blue) along with its ETFE counterpart (i.e. $|h_1(\mathbf{j}\omega)| = |\Delta\mathbf{i}_1(\mathbf{j}\omega)|/|\Delta\mathbf{v}_1(\mathbf{j}\omega)|$)}
    \label{fig:ETFE}
\end{figure}
\begin{figure}[h]
    \centering
    \includegraphics[scale=0.29]{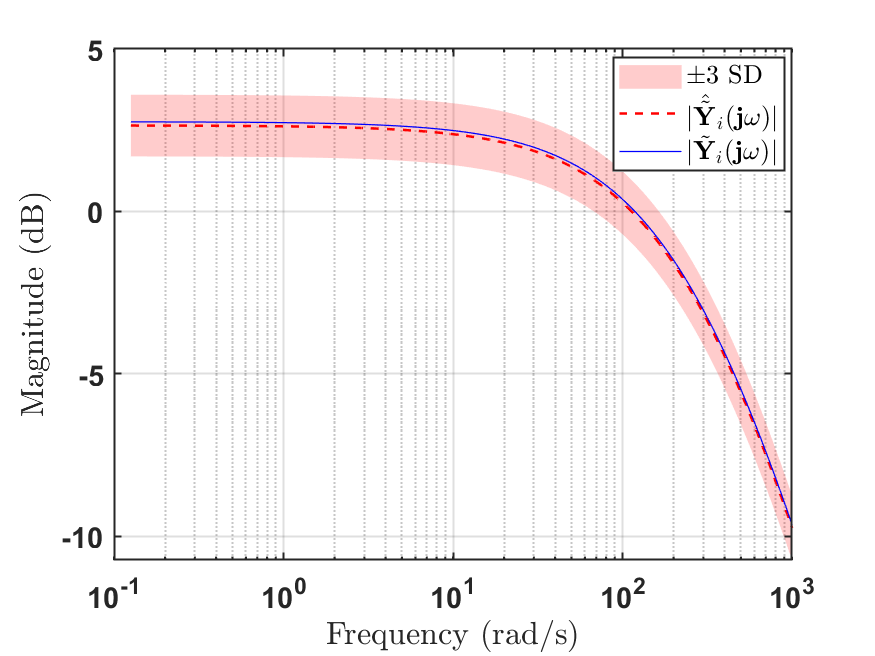}
    \includegraphics[scale=0.32]{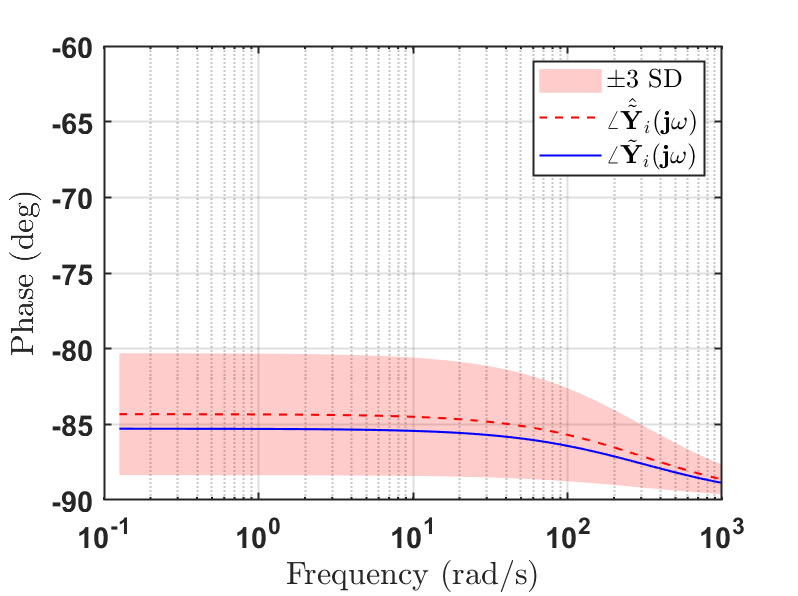}
    \caption{Representative plots of the estimated (in hyphenated red) (a) (left) Magnitude ($|\mathbf{\Tilde{Y}}_1(\mathbf{j}\omega)|$) and (b) (right) Phase ($\angle \mathbf{\Tilde{Y}}_1(\mathbf{j}\omega)$) along with the true values (in blue), across $20$ Monte-Carlo experiments. Note that the phase angle is non-zero at zero frequency owing to complex poles in $\mathbf{\Tilde{Y}}_1(\mathbf{j}\omega)$. The shaded region in red denotes the $\pm 3$ standard deviation across the experiments.}
\label{fig:parameterestimation_results}
\end{figure}
\begin{figure}[h]
    \centering
    \includegraphics[scale=0.32]{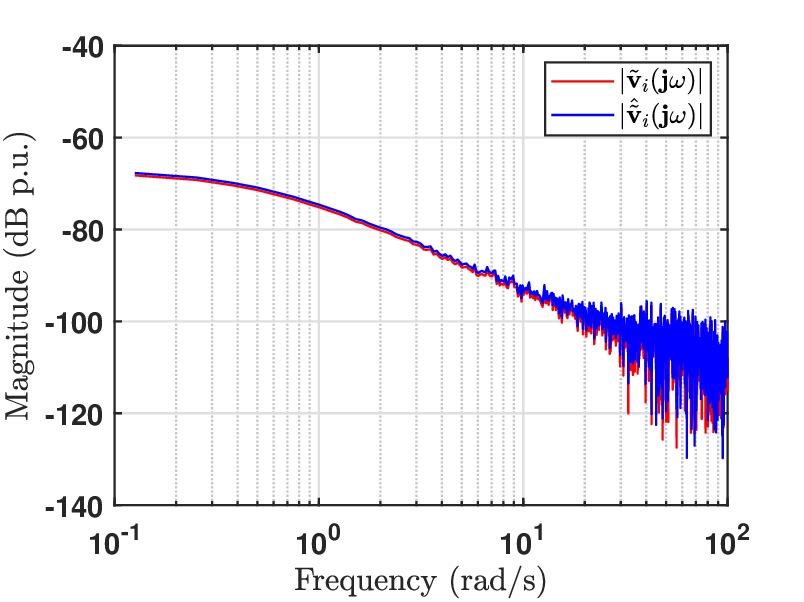}
    \includegraphics[scale=0.3]{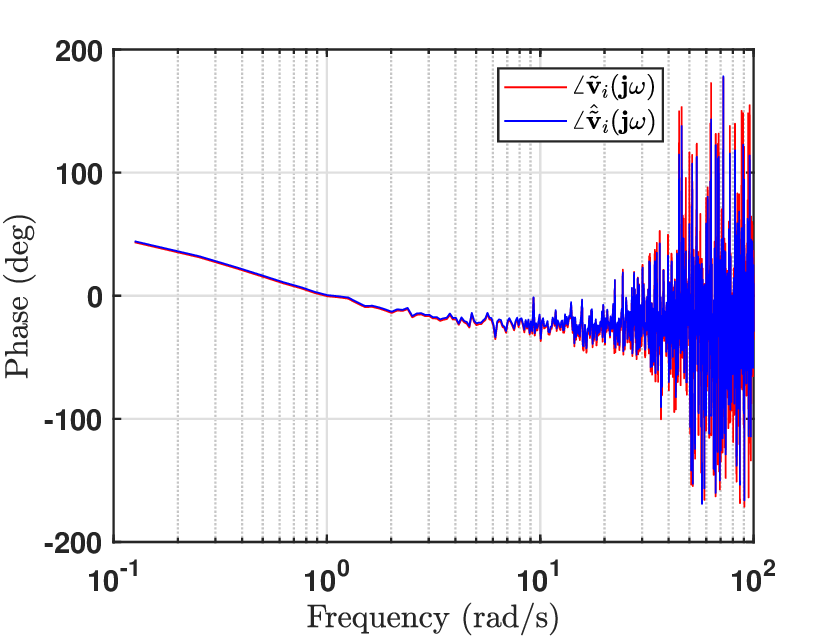}
    \caption{Representative plots of (a) (left) magnitude of the equivalent voltage $| \Tilde{\mathbf{v}}_i(\mathbf{j}\omega)|$ and (b) (right) its phase $\angle \Tilde{\mathbf{v}}_i(\mathbf{j}\omega)$, in the lower frequency bins where the active components are prominent. The red and blue plots denote the true and estimated values, respectively.}
    \label{fig:equivalentvoltage}
\end{figure}
\begin{figure}
    \centering
    \includegraphics[width=0.49\linewidth]{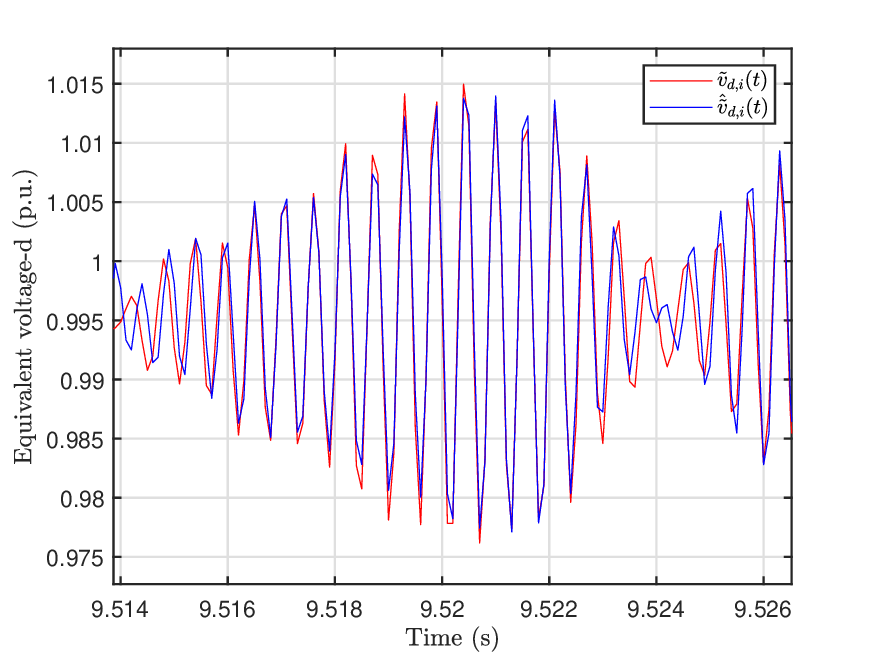}
    \includegraphics[width=0.49\linewidth]{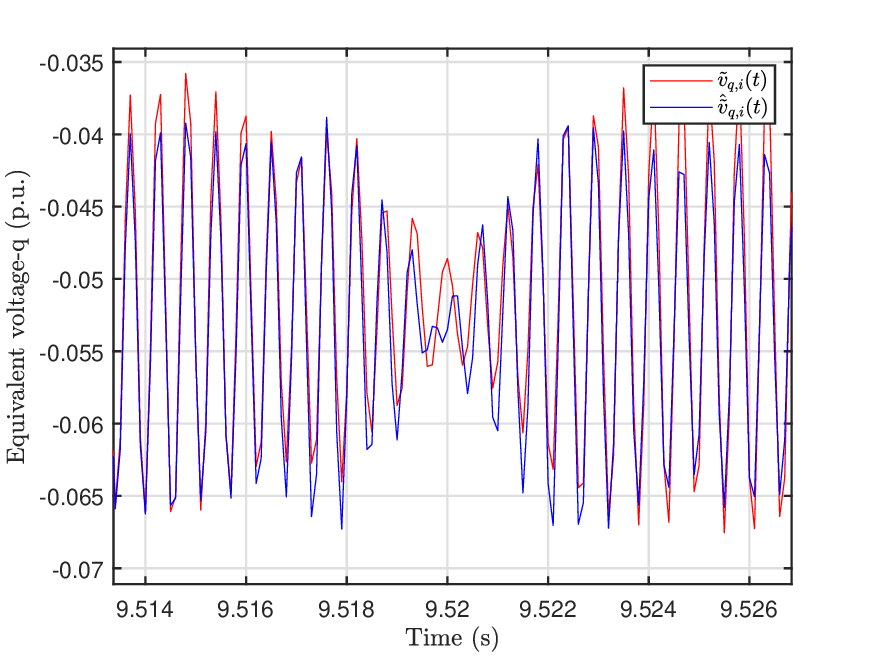}
    \caption{Representative plots of the true equivalent voltage (in red) along with the estimated values (in blue) in the time domain along (a) (left) $d$ and (b) (right) $q$ coordinates.}
    \label{fig:timedomainresult}
\end{figure}
\begin{figure}[h]
    \centering
    \includegraphics[scale=0.33]{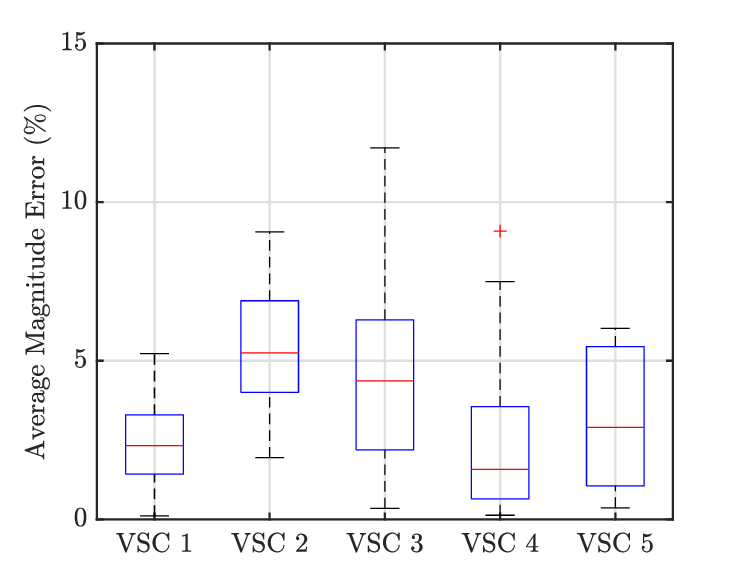}
    \includegraphics[scale=0.38]{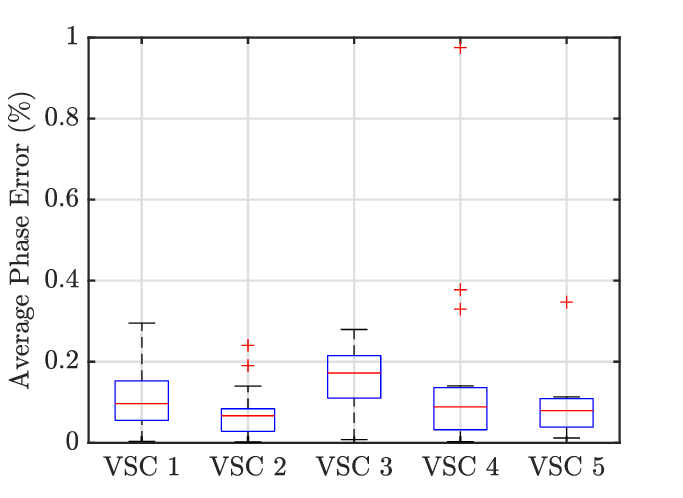}
    \caption{A box plot of (a) (left) the average equivalent admittance magnitude error and (b) (right) the average equivalent admittance phase error as seen from all the $5$ VSCs across $20$ Monte Carlo trials.}
    \label{fig:allVSCresult}
\end{figure}
\begin{figure}[h]
    \centering
    \includegraphics[scale=0.28]{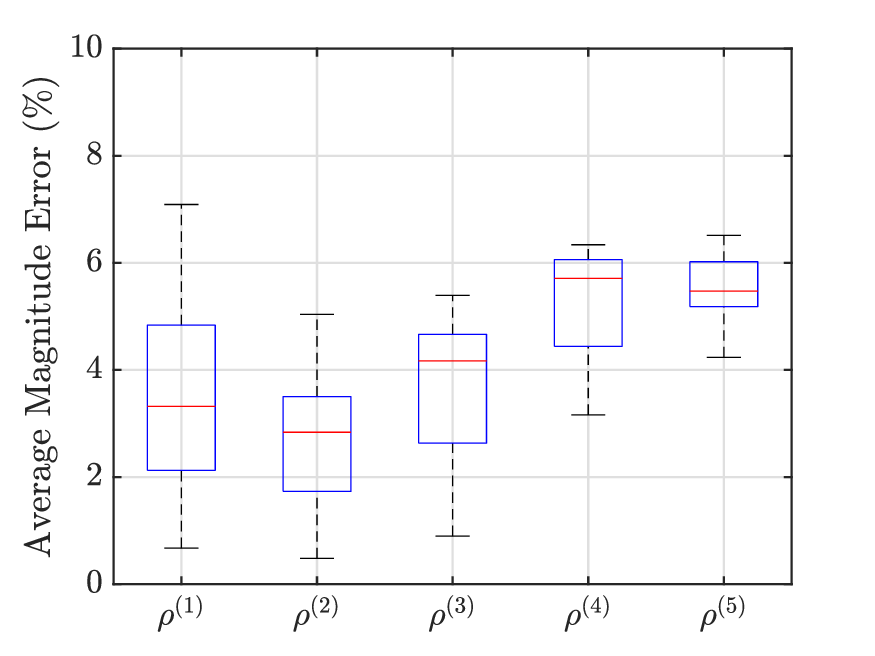}
    \includegraphics[scale=0.28]{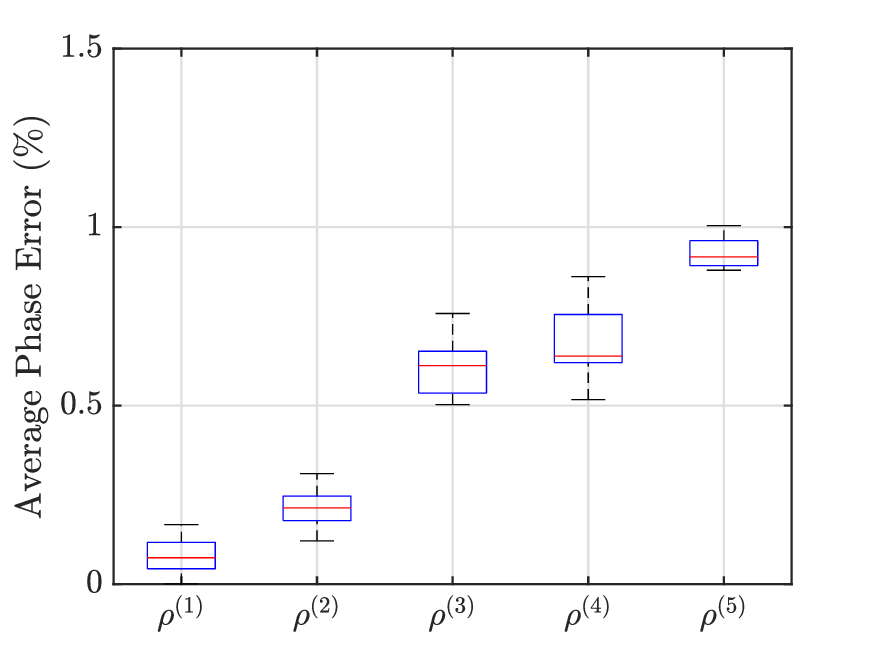}
    \caption{A box plot of (a) (left) the average equivalent admittance magnitude error and (b) (right) the average equivalent admittance phase error w.r.t. an increasing $\rho$ across $20$ Monte Carlo trials. The results shown here are as seen from VSC $1$.}
    \label{fig:increasingrho}
\end{figure}

We now describe the simulation results using our proposed algorithm.
Fig. \ref{fig:parameterestimation_results}(a),(b) show the resulting magnitude and phase plots for the equivalent transfer function estimated from the perspective of VSC 1 ($\Tilde{\mathbf{Y}}_1(s)$), using the algorithm in Section \ref{parameterest}, respectively. We observe an average magnitude and phase error of $0.3 \ \rm dB (3.45\%)$ and $0.1 \ \rm deg (0.11\%)$, with a maximum of $0.56 \ \rm dB (6.67\%)$ and $4.22 \ \rm deg (4.95\%)$, respectively.

Upon estimating the parameters $\hat{\theta}_i$, we then carry out equivalent voltage estimation as explained in Section \ref{disturbanceest}. Note that Section \ref{disturbanceest} describes the estimation for the "small-signal" equivalent voltage $\Delta \Tilde{\mathbf{v}}_i$. We add back the steady-state value and obtain $\Tilde{\mathbf{v}}_i =  \Delta \Tilde{\mathbf{v}}_i + \Tilde{\mathbf{v}}_i^{ss}$. Since $\Tilde{\mathbf{v}}_i^{ss}$ is unknown, we compute it using \eqref{thevenin}, i.e. $\hat{\Tilde{\mathbf{v}}}_i^{ss} = -\frac{1}{\hat{\Tilde{\mathbf{Y}}}_i(\mathbf{j}0)}(\mathbf{i}_i(\mathbf{j}0) -  \hat{\Tilde{\mathbf{Y}}}_i(\mathbf{j}0)\mathbf{v}_i(\mathbf{j}0))$ using the estimated $\hat{\Tilde{\mathbf{Y}}}_i$ from parameter estimation. We then show the results of the equivalent voltage magnitude and phase estimation in Fig. \ref{fig:equivalentvoltage}(a),(b), respectively. Note that we use the relation in \eqref{equivalentvoltage_full} to compute the true equivalent voltage rather than the approximation in \eqref{equivalentmodelsimple}. We also show the corresponding estimation results in the time domain by taking the inverse Fourier transform and separating into the $d$ and $q$ coordinates in Fig. \ref{fig:timedomainresult}(a) and (b), respectively.

We then show the result of equivalent admittance magnitude and phase estimation across all the $5$ VSCs in Fig. \ref{fig:allVSCresult}(a),(b), respectively. The relative average magnitude and phase errors (in $\%$) are computed as $\sum_{k=0}^{N_s}\frac{|\hat{\Tilde{\mathbf{Y}}}_i(\mathbf{j}\omega_k)|-|{\Tilde{\mathbf{Y}}}_i(\mathbf{j}\omega_k)|}{|{\Tilde{\mathbf{Y}}}_i(\mathbf{j}\omega_k)|} \times 100$ (in absolute values and not in $\mathrm{dB}$) and $\sum_{k=0}^{N_s}\frac{\angle\hat{\Tilde{\mathbf{Y}}}_i(\mathbf{j}\omega_k)-\angle {\Tilde{\mathbf{Y}}}_i(\mathbf{j}\omega_k)}{\angle {\Tilde{\mathbf{Y}}}_i(\mathbf{j}\omega_k)
} \times 100$  (in degrees). We are able to observe uniform error distributions across all the $5$ VSCs (with the same hyper-parameters for all the VSCs), indicating that our algorithm is robust and accurately tracks the true equivalent admittance withstanding simultaneous excitations from all the VSCs.

Another interesting result is the effect of the transmission line inductances on the estimation error. In Fig. \ref{fig:increasingrho}(a),(b), we show the resulting magnitude and phase estimation errors of the equivalent admittance associated with VSC $1$ upon varying the R-L ratio of all the lines, respectively. Here, $\rho^{(l)}$ refers to the set of all R-L ratios $\rho_{ij}$ in the range $(0.1(l-1), 0.1l)$. We see that despite deviating from Assumption \ref{RLassumption}(ii),(iii) (i.e., heterogeneous and not strongly inductive lines), we are still able to observe acceptable results on the admittance estimation. Further note that as the lines become less inductive (i.e., along the $x$ axis), the magnitude and phase error increases, thus empirically confirming one of the inferences from Theorem \ref{theorem1}.

\section{Conclusion}\label{conclusionsection}
\textcolor{black}{We proposed a parallel and decentralized algorithm in the frequency domain for grid-equivalent model estimation from the point of view of each converter. To separate the local equivalent impedance from that of the equivalent grid voltage, we designed a bilinear mapping framework that dissolves into two separate estimation problems for equivalent parameter and grid voltage estimation using frequency-domain pre-processing techniques. We then used two estimation algorithms, a least-squares algorithm and a Kalman filter, to estimate the parameters and the grid voltage, respectively, and showed promising results on an interconnected $5-$converter system, all of them in grid-forming mode. We also provided theoretical guarantees on the estimation error bounds on the parameter and equivalent voltage. Future compelling directions include applying the algorithm in a real-world setup; speeding up the estimation; stability analysis of the grid using the estimated model; including synchronous generators and other control techniques for converters (for example, virtual oscillator control, virtual synchronous machine, grid-following control schemes, etc.) in the simulations; identifying the equivalent inertia and damping from the equivalent voltage, and finally, designing decentralized adaptive control algorithms using the identified equivalent model.}

\section*{Appendix}
\subsection{Equivalent Admittance \& Voltage}\label{eqadmvoltcomp}
\textcolor{black}{Using the relation between the PCC and internal voltages $\Delta \mathbf{v}, \Delta \mathbf{u}$, respectively in the partitioned-dynamic grid model in \eqref{gridpartition}, we obtain the following:
\begin{align}
    \Delta \mathbf{i}_i &= \mathbf{Y}_{ii} \Delta \mathbf{v}_i + \mathbf{Y}_{i,-i}(\Delta \mathbf{u}_{-i} - \mathbf{Z}_{-i}^f\Delta \mathbf{i}_{-i}), \label{eqn1} \\
    \Delta \mathbf{i}_{-i} &= \mathbf{Y}_{i,-i}^\top \Delta \mathbf{v}_i + \mathbf{Y}_{-i,-i}(\Delta \mathbf{u}_{-i} - \mathbf{Z}_{-i}^f\Delta \mathbf{i}_{-i}), \label{eqn2}
\end{align}
where, ${\mathbf{Z}_{-i}^f} := \mathrm{diag}[\mathbf{Z}_1^f, \ldots, \mathbf{Z}_{i-1}^f, \mathbf{Z}_{i+1}^f, \ldots \mathbf{Z}_N^f] \in \mathbb{C}^{N-1\times N-1}$ is a diagonal matrix comprising the filter impedance of all the VSCs except $i$. Re-arranging the terms in \eqref{eqn2}, we obtain
\begin{equation*}
    \Delta \mathbf{i}_{-i} = (I_{N-1} + \mathbf{Y}_{-i,-i}\mathbf{Z}_{-i}^f)^{-1}(\mathbf{Y}_{i,-i}^\top \Delta \mathbf{v}_i + \mathbf{Y}_{-i,-i}\Delta \mathbf{u}_{-i}).
\end{equation*}
Upon inserting the above relation in \eqref{eqn1}, we obtain $\Delta \mathbf{i}_i = \Tilde{\mathbf{Y}}_i(\Delta \mathbf{v}_i - \Delta \Tilde{\mathbf{v}}_i)$, with
$\Tilde{\mathbf{Y}}_i, \Delta \Tilde{\mathbf{v}}_i$ in \eqref{equivalentadmittance_full}, \eqref{equivalentvoltage_full}, respectively.}

\subsection{Simplified Equivalent Admittance \& Voltage}\label{eqadmvoltcompsimple}
\textcolor{black}{Using Assumption \ref{RLassumption} and the equations in \eqref{simplifiedeqns} in the equivalent admittance in \eqref{equivalentadmittance_full}, we obtain the following relation:
\begin{align*}
    \Tilde{\mathbf{Y}}_i =& \frac{1}{s'+\rho}\sum_{m=1,m\neq i}^N \gamma_{im}  \\
    & -  \frac{s'+\rho^f}{(s'+\rho)^2}\Gamma_{i,-i}L_{-i}^f(I_{N-1} + \frac{s'+\rho^f}{s'+\rho}\Gamma_{-i,-i}L_{-i}^f)^{-1}\Gamma_{i,-i}^\top,
\end{align*}
with $s' = s + \mathbf{j}$. Note that with $s=\mathbf{j}\omega$ the fraction $\frac{s'+\rho^f}{s'+\rho}$ becomes $M(\omega)\exp(\mathbf{j\phi}(\omega))$, with the magnitude and phase being $M(\omega) = \sqrt{\frac{(\omega+1)^2 + {\rho^f}^2}{(\omega+1)^2 + \rho^2}},\phi = \arctan(\frac{(\omega+1)}{\rho^f}) - \arctan(\frac{(\omega+1)}{\rho})$, respectively. Since we have $\rho,\rho^f << 1$ from Assumption \ref{RLassumption}, $M (\omega)\approx 1, \phi(\omega) \approx 0, \forall \omega >0$, leading to $\frac{s'+\rho^f}{s'+\rho} \approx 1$. Thereby, using this approximation, we obtain a convenient first-order representation for the equivalent admittance in \eqref{equivalentmodelsimple}. The relation for the simplified equivalent voltage follows similarly.}

\subsection{Proof of Lemma \ref{heterolines}}\label{heteroproof}
\textcolor{black}{Let the operator $\mathcal{O}(s^{-r})$ represent a proper transfer function (scalar or matrix) of relative degree $r$. We now use only Assumption \ref{RLassumption}(i) to derive the relative degree of the equivalent admittance. Recall the equivalent admittance relation in \eqref{equivalentadmittance_full}.}

\textcolor{black}{Since each $\mathbf{Y}_{ij}$ is a strictly proper transfer function of order $1$, $\mathbf{Y}_{ii} \in \mathcal{O}(s^{-1})$. Similarly, $\mathbf{Y}_{i,-i}, \mathbf{Y}_{-i,-i} \in \mathcal{O}(s^{-1})$. Since the filters are all first-order, we have $\mathbf{Z}_{-i}^f \in \mathcal{O}(s^1)$.}

\textcolor{black}{Note that the relative degree of the multiplication of two or more transfer functions is the sum of their individual relative degrees. Thus, we have $\mathbf{Y}_{i,-i}\mathbf{Z}_{-i}^f \in \mathcal{O}(s^0)$ and $\mathbf{Y}_{-i,-i}\mathbf{Z}_{-i}^f \in \mathcal{O}(s^0)$. Also note that when you add two transfer functions of the same relative degree, the sum retains that relative degree. Therefore $I_{N-1} + Y_{-i,-i}\mathbf{Z}_{-i}^f \in \mathcal{O}(s^0)$. Since the matrix inverse of a proper transfer function matrix is also proper, we have $(I_{N-1} + Y_{-i,-i}\mathbf{Z}_{-i}^f)^{-1} \in \mathcal{O}(s^0)$. Therefore, we have $\mathbf{Y}_{i,-i}\mathbf{Z}_{-i}^f(I_{N-1} + Y_{-i,-i}\mathbf{Z}_{-i}^f)^{-1}\mathbf{Y}_{i,-i}^\top \in \mathcal{O}(s^{-1})$. Since we have already established that $Y_{ii} \in \mathcal{O}(s^{-1})$ , the overall equivalent admittance transfer function is strictly proper and has a relative degree $1$. A similar argument can be made for the equivalent voltage with $\frac{\mathbf{Y}_{i,-i}}{\mathbf{\Tilde{Y}}_i}( 1_{N-1} - \mathbf{Z}_{-i}^f(I_{N-1} + Y_{-i,-i}\mathbf{Z}_{-i}^f)^{-1}\mathbf{Y}_{-i,-i}) \in \mathcal{O}(s^{0})$}

\subsection{Proof of Theorem \ref{theorem1}}\label{thm1proof}
\textcolor{black}{The solution for the least-squares problem in \eqref{theta_opt} is given by $\hat{\theta}_i = (\sum_{k=0}^{N_s} H_i'(k)^\top H_i'(k))^{-1}(\sum_{k=0}^{N_s}H_i'(k)^\top z_i'(k))$. Note that $z_i'(k) = H_i'(k)\theta_i + \eta(k)$, $\eta(k) = -\Tilde{\gamma}_i\begin{bmatrix}
    \mathbb{Re}\{ \Tilde{h}_i'(k) \}\\
    \mathbb{I}\{ \Tilde{h}_i'(k) \}
\end{bmatrix}$, with $\Tilde{h}_i'(k) = S_i^{R\Tilde{V}}(k)/S_i^{RV}(k)$ (from \eqref{thevenin_IV}). This leads to the following relation:
\begin{equation}\label{thm1eqn1}
    \hat{\theta}_i = \theta_i + \underset{\epsilon_i^\theta}{\underbrace{(\sum_{k=0}^{N_s} H_i'(k)^\top H_i'(k))^{-1}(\sum_{k=0}^{N_s}H_i'(k)^\top \eta(k))}}.
\end{equation}
We will now establish an upper bound on the norm of the parameter estimation error $\epsilon_i^\theta$.}

\textcolor{black}{We rewrite the parameter estimation error as $\epsilon_i^\theta = (\mathbf{H}_i^\top\mathbf{H}_i)^{-1}\mathbf{H}_i^\top\mathbf{\eta}$, where $\mathbf{H}_i:= \begin{bmatrix}
    H_i'(0) & H_i'(1) & \ldots & H_i'(N_s)
\end{bmatrix}^\top$ and $\mathbf{\eta} := \begin{bmatrix}
    \eta(0) & \eta(1) & \ldots & \eta(N_s)
\end{bmatrix}^\top$, thus leading to the upper-bound
\begin{equation}\label{thm1eqn2}
    \lVert \epsilon_i^\theta \rVert \leq \lVert (\mathbf{H}_i^\top\mathbf{H}_i)^{-1}\mathbf{H}_i^\top \rVert \lVert \mathbf{\eta} \rVert.
\end{equation}}
\textcolor{black}{Note that
\begin{equation*}
    \lVert (\mathbf{H}_i^\top\mathbf{H}_i)^{-1}\mathbf{H}_i^\top \rVert = 1/\sigma_{\rm min}\{\mathbf{H}_i\} = 1/\sqrt{\lambda_{\rm min}\{\mathbf{H}_i^\top \mathbf{H}_i\}},
\end{equation*}
where $\sigma_{min}\{A\}$ and $\lambda_{min}\{A\}$ denote the minimum singular value and eigenvalue of a matrix $A$, respectively. According to the definition of $H_i'(k)$ in \eqref{theta_opt}, we have 
\begin{equation*}
    \mathbf{H}_i^\top\mathbf{H}_i = \begin{bmatrix}
        \sum_{k=0}^{N_s} |h_i'(k)|^2 & -\sum_{k=0}^{N_s}\mathbb{R}\{ h_i'(k) \}\\
        \sum_{k=0}^{N_s}\mathbb{R}\{ h_i'(k) \} & N_s
    \end{bmatrix},
\end{equation*}
whose trace and determinant are given by $\mathrm{Tr}\{\mathbf{H}_i^\top\mathbf{H}_i\} = \sum_{k=0}^{N_s} |h_i'(k)|^2 + N_s$, $\mathrm{D}\{\mathbf{H}_i^\top\mathbf{H}_i \} = N_s\sum_{k=0}^{N_s} |h_i'(k)|^2 - (\sum_{k=0}^{N_s}\mathbb{R}\{ h_i'(k) \})^2$, respectively.} \textcolor{black}{Since $\mathbf{H}_i^\top\mathbf{H}_i$ is positive definite, we have the maximum eigenvalue $\lambda_{\rm max}\{\mathbf{H}_i^\top \mathbf{H}_i\} < \mathrm{Tr}\{\mathbf{H}_i^\top\mathbf{H}_i\}$ and $\lambda_{\rm min}\{\mathbf{H}_i^\top \mathbf{H}_i\} = \mathrm{D}\{\mathbf{H}_i^\top\mathbf{H}_i \}/\lambda_{\rm max}\{\mathbf{H}_i^\top \mathbf{H}_i\} > \mathrm{D}\{\mathbf{H}_i^\top\mathbf{H}_i \}/\mathrm{Tr}\{\mathbf{H}_i^\top\mathbf{H}_i \}$ leading to $1/\lambda_{\rm min}\{\mathbf{H}_i^\top\mathbf{H}_i\} < \mathrm{Tr}\{\mathbf{H}_i^\top\mathbf{H}_i\}/\mathrm{D}\{\mathbf{H}_i^\top\mathbf{H}_i\}$. Therefore, we have
\begin{equation}\label{thm1eqn3}
    \lVert (\mathbf{H}_i^\top\mathbf{H}_i)^{-1}\mathbf{H}_i^\top \rVert < \sqrt{\frac{N_s + \sum_{k=0}^{N_s} |h_i'(k)|^2}{N_s\sum_{k=0}^{N_s} |h_i'(k)|^2 - (\sum_{k=0}^{N_s}\mathbb{R}\{ h_i'(k) \})^2}}.
\end{equation}
We now look at the term $\lVert \mathbf{\eta}\rVert$: $\lVert \mathbf{\eta}\rVert \leq \Tilde{\gamma}_i\sqrt{\sum_{k=0}^{N_s} |\Tilde{h}_i'(k)|^2}$. Recall the assumption in Theorem \ref{theorem1}:
\begin{align*}
    \mathbf{v}_i(\mathbf{j}\omega) &= G_i(\mathbf{j}\omega)\mathbf{r}_i(\mathbf{j}\omega) + \sum_{j=1,j\neq i}^{N} G_j(\mathbf{j}\omega)\mathbf{r}_j(\mathbf{j}\omega) + n_i(\mathbf{j}\omega),\\
    \Tilde{\mathbf{v}}_i(\mathbf{j}\omega) &= \Tilde{G}_i(\mathbf{j}\omega)\mathbf{r}_i(\mathbf{j}\omega) + \sum_{j=1,j\neq i}^{N} \Tilde{G}_j(\mathbf{j}\omega)\mathbf{r}_j(\mathbf{j}\omega) + \Tilde{n}_i(\mathbf{j}\omega).
\end{align*}
Using the IV projection in \eqref{thevenin_IV}, we have $h_i'(k) = \Tilde{G}_i(k)/G_i(k)$, which represents the dynamic coupling between the local PCC $i$ and the equivalent grid at frequency $\omega_k$.}
 
\textcolor{black}{We note that ${\sum_{k=0}^{N_s} |\Tilde{h}_i'(k)|^2} = N_s\mathrm{Mean}\{|\Tilde{h}_i'(k)|^2\}$. By considering $\mathcal{S}_i = \sqrt{\text{Mean}\{\vert{}\tilde{G}_i(k)/G_i(k)\vert{}^2\}}$, we arrive at $\lVert \eta \rVert \leq \Tilde{\gamma}_i \mathcal{S}_i\sqrt{N_s}$. Using this inequality and \eqref{thm1eqn3} in \eqref{thm1eqn2}, we obtain:
\begin{align}
    \lVert \epsilon_i^\theta \rVert & < \Tilde{\gamma}_i \mathcal{S}_i\sqrt{N_s} \sqrt{\frac{N_s + \sum_{k=0}^{N_s} |h_i'(k)|^2}{N_s\sum_{k=0}^{N_s} |h_i'(k)|^2 - (\sum_{k=0}^{N_s}\mathbb{R}\{ h_i'(k) \})^2}}, \nonumber \\
    &  = \Tilde{\gamma}_i \mathcal{S}_i\sqrt{N_s} \sqrt{\frac{1 + \frac{1}{N_s}\sum_{k=0}^{N_s} |h_i'(k)|^2}{\sum_{k=0}^{N_s} |h_i'(k)|^2 - \frac{1}{N_s}(\sum_{k=0}^{N_s}\mathbb{R}\{ h_i'(k) \})^2}}, \nonumber \\
    &  = \Tilde{\gamma}_i \mathcal{S}_i \sqrt{\frac{1 + \frac{1}{N_s}\sum_{k=0}^{N_s} |h_i'(k)|^2}{\frac{1}{N_s}\sum_{k=0}^{N_s} |h_i'(k)|^2 - (\frac{1}{N_s}\sum_{k=0}^{N_s}\mathbb{R}\{ h_i'(k) \})^2}}, \nonumber\\
    & = \Tilde{\gamma}_i \mathcal{S}_i \sqrt{\frac{1 + \mathrm{Mean}\{ |h_i'(k)|^2\}}{\mathrm{Mean}\{ |h_i'(k)|^2\} - ( \mathrm{Mean}\{\mathbb{R}\{ h_i'(k) \} \})^2}}, \nonumber\\
    & = \Tilde{\gamma}_i \mathcal{S}_i \sqrt{\frac{1 + \mathrm{Mean}\{ |h_i'(k)|^2\}}{\mathrm{Var}\{ \mathbb{R}\{ h_i'(k) \}\} + \mathrm{Mean}\{\mathbb{I}\{ h_i'(k) \}^2 \}}}. \label{thm1finaleqn}
\end{align}
By replacing $k$ with $\omega_k$ in the above equation, we obtain the main result in Theorem \ref{theorem1}.} 

\subsection{Proof of Theorem \ref{theorem2}}\label{thm2proof}
\textcolor{black}{Recall the bilinear map from \eqref{maineqn}: $z_i(k) = H_i(k)\theta_i + [\theta_i]_2 d_i(k)$. The residual for the Kalman filter algorithm in Section \ref{disturbanceest} is given by  $\tilde{z}_i(k) = z_i(k) - H_i(k)\hat{\theta}_i$. Substituting $\theta_i = \hat{\theta}_i - \epsilon_i^\theta$ yields
\begin{equation*}
    \tilde{z}_i(k) = \left( [\hat{\theta}_i]_2 - [\epsilon_i^\theta]_2 \right) d_i(k) - H_i(k)\epsilon_i^\theta.
\end{equation*}
Using the above residual in the posterior update equation \eqref{dupdate} and then subtracting $d_i(k)$ on both sides leads to:
\begin{align}
    \epsilon_i^d(k) =& \hat{d}_i(k-1) - d_i(k) + \nonumber\\
    &K_d(k) \Big( [\hat{\theta}_i]_2(d_i(k) - \hat{d}_i(k-1)) - [\epsilon_i^\theta]_2 d_i^*(k) - H_i(k)\epsilon_i^\theta \Big) \label{thm2eqn1},
\end{align}
where, the error for the estimation of $d_i$ at frequency $\omega_k$ is given by $\epsilon_i^d(k):= \hat{d}_i(k) - d_i(k)$. Let $\Delta d_i(k) := d_i(k) - d_i(k-1)$ denote the difference of the true value of $d$ between consecutive frequency data points. By writing $ \hat{d}_i(k-1)-d_i(k) = \epsilon_i^d(k-1) - \Delta d_i(k)$ and substituting this in \eqref{thm2eqn1}, we obtain:}

\footnotesize
\textcolor{black}{\begin{align}
    \epsilon_i^d(k) = & \Phi(k)\big( \epsilon_i^d(k-1) - \Delta d_i(k) \big) - K_d(k)\big( H_i(k)\epsilon_i^\theta + [\epsilon_i^\theta]_2 d_i(k) \big),
\end{align}
\normalsize where, $\Phi(k) = I_2 - K_d(k)[\hat{\theta}_i]_2$. Using the Kalman filter gain in \eqref{KFgain} and applying the Woodbury matrix lemma, we obtain $\Phi(k) = \left( I_2 + [\hat{\theta}_i]_2^2 (P_d(k)+Q) R_d^{-1} \right)^{-1}$. Since $Q= \sigma_q I_2$ and $R_d = c_1 I_2$, with $\sigma_q,c_1 >0$ (refer Section \ref{disturbanceest}), and $P_d(\cdot)$ is positive definite, there exists a $\varphi \in (0,1)$ such that
\begin{equation}\label{thm2eqn2}
    \Vert{}\Phi(k)\Vert{} \le \frac{1}{1 + \frac{[\hat{\theta}_i]_2^2}{c_1} \lambda_{\min}\{P_d(k)+Q\}} \leq\varphi < 1, \forall \ k \in \{0,\ldots,N_s\}.
\end{equation} Also note that the Kalman filter gain $K_d(k)$ is symmetric, leading to 
\begin{equation}\label{thm2eqn3}
    \lVert K_d(k) \rVert \leq \lambda_{\rm max}\{ K_d(k) \} \leq 1/[\hat{\theta}_i]_2^2.
\end{equation}}
\textcolor{black}{\normalsize Upon taking the expected squared norm in \eqref{thm2eqn1} and updating it with the identities \eqref{thm2eqn2} and \eqref{thm2eqn3}, we obtain:
\begin{align}
    \mathbb{E}[\Vert{}\epsilon_i^d(k)\Vert{}^2] \le & \varphi^2 \Big( \mathbb{E}[\Vert{}\epsilon_i^d(k-1)\Vert{}^2]
    +\mathbb{E}[\Vert{}\Delta d_i(k)\Vert{}^2] \Big) + \nonumber \\
     &\frac{1}{[\hat{\theta}_i]_2^2} \mathbb{E} \Big[ \big\Vert{} H_i(k)\epsilon_i^\theta + [\epsilon_i^\theta]_2 d_i(k) \big\Vert{}^2 \Big], \label{thm2eqn5}
\end{align}
where, we made use of the property $\lVert A  B \rVert \leq \lVert A \rVert \lVert B \rVert$ and the fact that the expectation of the cross terms are zero.}

\textcolor{black}{\normalsize When we note the structure of $H_i(k)$ in \eqref{maineqn}, we notice that the trace of $H_i(k)^\top H_i(k)$ is $1 + |h_i(k)|^2$. Furthermore, since $H_i(k)^\top H_i(k)$ is positive-definite, we have $\lVert H_i(k) \rVert^2 \leq 1 + |h_i(k)|^2$. Updating this in \eqref{thm2eqn5}, we get
\begin{align}
    \mathbb{E}[\Vert{}\epsilon_i^d(k)\Vert{}^2] \le & \varphi^2 \Big( \mathbb{E}[\Vert{}\epsilon_i^d(k-1)\Vert{}^2]
    +\mathbb{E}[\Vert{}\Delta d_i(k)\Vert{}^2] \Big) + \nonumber \\
     &\frac{2\lVert \epsilon_i^\theta\rVert^2}{[\hat{\theta}_i]_2^2}(1 + |h_i(k)|^2 + \mathbb{E} \Big[\lVert d_i(k) \big\Vert{}^2 \Big])  ,\label{thm2eqn6}
\end{align}
\normalsize where we used the properties $\lVert A + B \rVert^2 \leq 2(\lVert A \rVert^2 + \lVert B \rVert^2)$ and $\lVert [\epsilon_i^\theta]_2 \rVert \leq \lVert \epsilon_i^\theta \rVert$. By replacing $k$ with $\omega_k$, we obtain the main result in Theorem \ref{theorem2}.}

\bibliographystyle{IEEEtran}

\end{document}